%% file: main.tex
\documentclass[11pt]{article}

\usepackage{amssymb,amsmath,amsthm,amsfonts}

\usepackage[T1]{fontenc}
\usepackage[tt=false, type1=true]{libertine}
\usepackage[varqu]{zi4}
\usepackage[libertine]{newtxmath}

\usepackage[margin=1in]{geometry}
\usepackage{enumitem}
\setlist{nosep,topsep=0pt,leftmargin=*}

\usepackage{hyperref}
\hypersetup{
    colorlinks,
    allcolors=myblue
}

\usepackage[sortcites,sorting=nyt,style=alphabetic]{biblatex}
\input{commands}

\newtheorem{theorem}{Theorem}[section]
\newtheorem{corollary}[theorem]{Corollary}
\newtheorem{lemma}[theorem]{Lemma}

\newtheorem{proposition}[theorem]{Proposition}

\theoremstyle{definition}

\newtheorem{definition}{Definition}[section]
\newtheorem{assumption}{Assumption}[section]

\newcommand{\citet}[1]{\textcite{#1}}

\title{No-Regret Algorithms in non-Truthful Auctions with Budget and ROI Constraints}

\author{
    Gagan Aggarwal\\
    Google Research\\
    \texttt{gagana@google.com}%
    \and
    Giannis Fikioris%
    \thanks{Supported in part by the Department of Defense (DoD) through the National Defense Science \& Engineering Graduate (NDSEG) Fellowship, the Onassis Foundation -- Scholarship ID: F ZS 068-1/2022-2023, and AFOSR grant FA9550-23-1-0068.}\\
    Cornell University\thanks{Part of this work was done while Giannis Fikioris was visiting Google as a Student Researcher in the Market Algorithms team.}\\
    \texttt{gfikioris@cs.cornell.edu}%
    \and
    Mingfei Zhao\\
    Google Research\\
    \texttt{mingfei@google.com}%
}

\date{\vspace{-25pt}}

\begin{document}

\maketitle{}

\begin{abstract}
    \input{body/0.abstract}
\end{abstract}

\input{body/1.intro}
\input{body/2.prelim}
\input{body/3.full}

\input{body/4.tightROI}
\input{body/5.bandit}
\input{body/6.poly}

\printbibliography{}

\appendix{}
\input{appendix/2.prelim}

\input{appendix/3.full}
\input{appendix/4.tight}
\input{appendix/5.bandit}
\input{appendix/6.poly}

\end{document}

%% file: commands.tex
\usepackage[ruled]{algorithm2e} % For algorithms
    
    \SetAlFnt{\small}
    \SetAlCapFnt{\small}
    \SetAlCapNameFnt{\small}
    \SetAlCapHSkip{0pt} 
    \IncMargin{-\parindent}

\usepackage{xspace,nicefrac,tcolorbox,xifthen,tikz,subcaption,enumitem,physics}
\usetikzlibrary{shapes, arrows, positioning}

\usepackage[bb=dsserif]{mathalpha}

\usepackage{xcolor}
    \definecolor{myblue}{HTML}{4285f4}

\usepackage[capitalize]{cleveref}
\renewcommand{\comment}[1]{}
\renewcommand{\paragraph}[1]{\smallskip\noindent\textit{#1.}}
\newcommand{\boldparagraph}[1]{\smallskip\noindent\textbf{#1.}}

\let\a\alpha
\let\b\beta
\let\g\gamma
\let\e\varepsilon
\let\d\delta
\let\l\lambda
\let\s\sigma
\let\sub\subseteq

\newcommand{\N}{\mathbb{N}}
\newcommand{\R}{\mathbb{R}}

\DeclareMathOperator*{\argmax}{\arg\max}

\DeclareMathOperator*{\E}{\mathbb{E}}
\DeclareMathOperator*{\PR}{\mathbb{P}}
\newcommand{\Es}{\E\nolimits}
\newcommand{\PRs}{\PR\nolimits}

\newcommand{\Ex}[2][]{\E_{#1}\left[#2\right]}
\newcommand{\Exs}[2][]{\Es_{#1}\left[#2\right]}
\newcommand{\ExC}[3][]{\E_{#1}\left[#2\phantom{\Big|}\middle|\phantom{\Big|}#3\right]}
\renewcommand{\Pr}[2][]{\PR_{#1}\left[#2\right]}
\newcommand{\Prs}[2][]{\PRs_{#1}\left[#2\right]}

\newcommand{\One}[1]{\mathbb{1}\left[#1\right]}
\usepackage{xifthen}
\newcommand{\Line}[4]{%
    #1&%
    \;\ifthenelse{\isempty{#2}}{\phantom{=}}{#2}\;%
    #3%
    \ifthenelse{\isempty{#4}}{}{&&\qquad\left(#4\big.\right)}%
}

\newcommand{\term}[1]{\ensuremath{\mathtt{#1}}}
\newcommand{\opt}{\term{OPT}}
\newcommand{\reg}{\term{Reg}}

\newcommand{\regP}{\reg_{\term{P}}}
\newcommand{\calF}{\mathcal{F}}
\newcommand{\calC}{\mathcal{C}}
\newcommand{\calL}{\mathcal{L}}
\newcommand{\calA}{\mathcal{A}}
\newcommand{\calD}{\mathcal{D}}
\newcommand{\calT}{\mathcal{T}}
\newcommand{\Clip}{C_{\mathtt{Lip}}}
\newcommand{\btt}{\mathtt{b}}
\newcommand{\KL}[2]{D_{\mathtt{KL}}\left(#1 \, \middle\Vert \, #2\right)}

\newcommand{\calFmul}{\calF_{\term{mul}}}
\newcommand{\calFlip}{\calF_{\term{Lip}}}

%% file: body/0.abstract.tex
Advertisers are increasingly using automated bidding to optimize their ad campaigns on online advertising platforms. Autobidding allows an advertiser to optimize her objective subject to various constraints, e.g. average ROI and/or budget constraints. In this paper, we study the problem of designing online autobidding algorithms to optimize value subject to ROI and budget constraints when the platform is running a first price auction or a mixture of first and second price auctions.  

We consider the following stochastic setting: There is one item for sale in each of $T$ rounds. In each round, the buyers submit their bids and an auction is run to sell the item. We focus on the bidding problem of one buyer, possibly with budget and ROI constraints. We assume that the buyer's value and the highest competing bid are drawn i.i.d. from some unknown (joint) distribution in each round. Our goal is to design a low-regret bidding algorithm to submit per-round bids on behalf of this buyer such that the buyer's constraints are satisfied. Our benchmark is the objective value achievable by the best possible Lipschitz function that maps values to bids, which is rich enough to best respond to many different correlation structures between value and highest competing bid, e.g. positive or negative correlation.
Our main result is an algorithm with full information feedback (i.e. the bidder observes the highest competing bid after each round) that guarantees a near-optimal $\tilde O(\sqrt T)$ regret with respect to the best Lipschitz function that maps values to bids. Our result applies to a wide range of auctions, most notably including any mixture of first and second price auctions (where the price is a convex combination of the first and second price). In addition, our result holds for both value-maximizing buyers and quasi-linear utility-maximizing buyers.

We also study the bandit setting, where the algorithm only observes whether the bidder wins the auction or not. In this setting, we show an $\Omega(T^{2/3})$ lower bound on the regret for first-price auctions, showing a large disparity between the full information and bandit settings. We also design an algorithm with a regret bound of $\tilde O(T^{3/4})$, when the value distribution is known and is independent of the highest competing bid.

%% file: body/1.intro.tex
\section{Introduction}

With the growth of online markets in terms of both complexity and scale, advertisers are increasingly turning towards autobidding to optimize their ad campaigns on online advertising platforms. Autobidding allows an advertiser to use an optimization algorithm to generate bids on her behalf, rather than manually bidding for each ad query. The advertiser provides high-level goals and constraints to the autobidder, which bids on her behalf in order to optimize her objective, while satisfying her constraints.

In this paper, we study the problem of designing algorithms for autobidding on behalf of a buyer. We consider a stochastic setting with $T$ rounds, in each of which one item is sold via an auction. In each round, the information of this round, including the buyer’s value and the highest competing bid, are drawn i.i.d. from some unknown (joint) distribution. The autobidder submits a bid to the auction based on her value and the history. If the bid is at least the highest competing bid, the bidder wins the current round and pays a price. The bidder has a budget constraint that limits the total payment, as well as a Return-on-Investment (ROI) constraint which requires that the total value in the winning rounds is at least a fraction of the total payment. The goal is to design an online bidding algorithm to maximize the bidder's objective subject to both budget and ROI constraints.
To quantify an algorithm's performance, we use (additive) regret against the objective value obtained by the best bidding strategy that knows the underlying distribution.

There has been a lot of recent work on the problem of designing algorithms for autobidding in stochastic settings. One line of work addresses the problem of designing no-regret online bidding algorithms for buyers with budget constraints~\cite{DBLP:conf/sigecom/BalseiroG17} or with both budget and ROI constraints~\cite{DBLP:conf/www/FengPW23} {\em when the seller is using a truthful} auction, e.g. second price.
These works use the best bidding sequence as the benchmark. Another line of work studies bidding in non-truthful auctions. Bidding in a non-truthful auction is usually harder than bidding in a truthful one, and this is reflected in the weaker benchmark used by these papers – the best constant pacing (also sometimes called uniform bidding), where the bid in each round is the true value multiplied by a constant multiplier~\cite{DBLP:conf/sigecom/FikiorisT23, DBLP:journals/corr/LucierPSZ23}. \cite{DBLP:journals/corr/CastiglioniCK23} study a similar setting with a stronger benchmark -- the best bid per value -- but let their bounds depend on the number of possible values and bids that the bidder can use, which in general can be uncountably many.

In this paper, we study the problem of bidding in non-truthful auctions and design no-regret algorithms against a stronger benchmark than the best constant pacing -- we present algorithms that have low regret compared to the {\em best Lipschitz bidding function that maps values to bids.}
Due to the generality of Lipschitz functions this benchmark can best-respond to a range of different correlations between the buyer's value and the highest competing bid, e.g. positive correlation for some values and negative correlation for others.

%In contrast to second-price auctions, sequential first-price auctions have not received as much attention however.
%\cite{DBLP:conf/sigecom/FikiorisT23} offer a no-regret algorithm when the player is only budget constrained but it is against a weak benchmark.
%\cite{DBLP:journals/corr/CastiglioniCK23} and \cite{DBLP:journals/corr/LucierPSZ23} offer no-regret algorithms for both budget and ROI constraints but work if the number of bids and values are finite or against the aforementioned weaker benchmark.
%This leaves a big gap for designing such algorithms for general sequential first-price auctions.
%The importance of this gap is amplified by the fact that many online platforms are switching to first-price auctions, e.g. \todo{citations}.

\boldparagraph{Our results and techniques}
We first consider the full-information setting where the bidder observes the highest-competing bid at the end of each round. 
We prove that there is an algorithm that can get near-optimal regret with respect to the best Lipschitz bidding function, i.e., it achieves an objective value that is close to the one achieved by using any Lipschitz continuous function that maps values to bids.
The main result for this setting is as follows:

\begin{theorem}[Informal version of \Cref{thm:tight:main_tight}]\label{thm:full_in_intro}
    There is an algorithm that achieves $\tilde O(\sqrt{T})$\footnote{$\tilde O(\sqrt{T}) = O(\sqrt{T}\cdot \text{poly}(\log(T)))$.} regret while satisfying both the budget and ROI constraints, with respect to the best Lipschitz bidding function given the knowledge of the distribution. The result applies to various classes of auctions (see \Cref{asmp:full}) including first-price auctions, second-price auctions and a hybrid of both. The result {also} applies to bidders with a value-maximizing or quasi-linear utility-maximizing objective.
\end{theorem}

% \gfcomment{The actual benchmark we use is the optimal mixture of Lipschitz bidding functions but I think this is too complicated for the introduction}
% \gacomment{yeah, I think we can leave that out of the intro}
%\todo{mention that it is an exponential time algorithm}

To the best of our knowledge, this is the first algorithm that achieves near-optimal regret bounds against the best Lipschitz bidding function for non-truthful auctions under budget and/or ROI constraints.\footnote{See \Cref{sec:related} for comparisons with prior works.}
Our result applies to any input distribution under mild assumptions (see \Cref{sec:prelim}).

Our algorithmic framework is based on the primal/dual framework for online learning with constraints \cite{DBLP:conf/sigecom/BalseiroG17,DBLP:conf/focs/ImmorlicaSSS19,DBLP:conf/icml/CastiglioniCK22,DBLP:conf/colt/FikiorisT23,DBLP:journals/corr/CastiglioniCK23,DBLP:conf/icml/BalseiroLM20}.
In this framework, to manage the global constraints, the `core' algorithm deploys two algorithms against each other.
Both algorithms aim to optimize an objective without considering the constraints.
On the one hand, the primal algorithm picks an action (which is subsequently used in the actual online learning problem) to maximize a function that comes from the Lagrangian of the problem.
On the other hand, the dual algorithm picks Lagrangian multipliers to minimize the same function.
The two algorithms participate in this sequential unconstrained stochastic zero-sum game, which implies the guarantees for the original constrained problem.

While the dual algorithm uses a standard instance of Online Gradient Descent to pick the scalars that represent the Lagrangian multipliers, designing the primal algorithm is often much more complicated and requires knowledge specific to the original problem.
%\gfreplace{Thus, given}{While the above discussion focuses on second-price auctions, it illustrates} the versatility and strength of this primal/dual paradigm\gfedit{. Thus}, research often focuses on designing good primal algorithms for the problem at hand. 
We develop the primal algorithm for our main result in \cref{sec:full}.

\boldparagraph{Main Technical Challenges}
Below we list some of the main technical challenges that we need to tackle and give a brief outline of our approach to solving them.

\paragraph{Lagrangian Maximization in Non-Truthful Auctions}
To better explain the challenge, we first consider the problem where the auction used is a second-price (or more generally truthful) auction. The part of the Lagrangian function that depends on the primal algorithm's bid $b$ takes the following form (for either value or quasi-linear utility maximization):
\begin{equation*}
    r(b)
    =
    \One{b \ge d} ( \chi v - \psi d)
    ,
\end{equation*}
where $v$ is the player's value, $d$ is the (unknown) highest competing bid, and $\chi,\psi$ are arbitrary non-negative numbers that depend on the Lagrange multipliers.
Maximizing the above function is straightforward: using $b^* = v \frac{\chi}{\psi}$ implies\footnote{We denote $x^+ = \max\{0, x\}$.} $r(b^*) = ( \chi v - \psi d )^+$, which guarantees maximum reward.
Since $b^*$ does not require knowledge of the highest competing bid $d$, the primal algorithm can pick this bid to guarantee zero regret for maximizing the Lagrangian.
In combination with the low regret guarantee of the dual algorithm, this satisfies the conditions of the primal/dual framework, leading to low regret guarantees for the original problem with constraints.

In contrast to the truthful auction setting where the best bid in every round is independent of the highest competing bid, for non-truthful auctions (e.g. first-price), the best bid in a round that maximizes the Lagrangian is a function of the value, the Lagrangian multipliers, and the highest competing bid.
Since the highest competing bid is unknown, the learner needs to learn the best function that maps values to bids and yields high reward.
However, this learning task is unrealistic, since the best such function might be non-monotone and discontinuous.
A more reasonable goal is to focus on a class of functions that have some structure.
Such a class used in previous work is the class of pacing multipliers, that map values to bids by multiplying by some constant number.
Instead, we focus on the more general class of Lipschitz continuous functions{, which can provide a much stronger benchmark to compete against, even in very simple settings where values and highest competing bids are independent.
For example, if the highest competing bid is constant and the value is not, the best response is a fixed bid, which cannot be expressed by the class of pacing multipliers.
More specifically, consider a value maximizer in first-price auctions with total budget $T/2$ and values that are $1/2$ or $1$, each with probability $1/2$.
Also assume that the highest competing bid every round is $1/2$.
The best Lipschitz bidding function is to bid $1/2$ every round and win all rounds. In contrast, no fixed pacing multiplier can win all rounds and stays within budget, which implies that the two benchmarks differ by a multiplicative factor.}

{While the above example is very simple, it showcases how much weaker the class of pacing multipliers is.
Another case that illustrates their simplicity is that in first-price auctions, an online learner who wants to compete against the best pacing multiplier does not have to consider the ROI constraint, since the optimal such bidding would never use pacing multipliers that would violate that constraint.
Despite the increased complexity of the class of Lipschitz functions, our $\tilde O(\sqrt T)$ regret cannot be improved even if we consider the weaker class}.

Both the class of pacing multipliers, $\calFmul$, and the class of Lipschitz continuous functions, $\calFlip$, have infinite cardinality.
However, the difference between the two becomes apparent when approximating them with a finite subset.
More specifically, $\calFmul$ can be approximated with accuracy $\e$ using a set of size $\Theta(1/\e)$.
In contrast, the same approximation for $\calFlip$ requires a set of size $\exp(\Theta(1/\e))$.
This fact prohibits using the following simple technique for $\calFlip$: discretize the class within accuracy $\e$ and run a standard no-regret algorithm with the finite class.
For $\calFmul$ this yields regret $O(T \e + \sqrt{T \log(1/\e)})$ ($T \e$ is the discretization error and $\sqrt{T \log K}$ is the regret of using $K$ different actions), but for $\calFlip$ this yields regret $O(T \e + \sqrt{T \log(\exp(1/\e))})$.
Picking the optimal $\e$ for each case, for $\calFmul$ we get $\tilde{O}(\sqrt T)$ regret but for $\calFlip$ we get $\tilde{O}(T^{2/3})$ regret, which is suboptimal.

To achieve the near-optimal $\tilde O(\sqrt T)$ regret, we use the structure implied by the finite subset of $\calFlip$.
More specifically, we create a tree where the functions of the finite subset of $\calFlip$ are the leaves and smaller distance between two leaves implies more similarity between the two corresponding functions. This allows us to enhance the standard regret guarantees of learning algorithms to get the improved result.
This tree algorithm can be found in \cref{ssec:full:tree}.

\paragraph{No-Regret Primal Algorithm against Adaptive Adversary and Time-Varying Range} The Lagrangian function that the primal algorithm aims to maximize depends on the Lagrangian multipliers picked by the dual algorithm.
This means that the primal algorithm's guarantees need to hold even against an adaptive adversary since no assumptions can be made for the dual algorithm's behavior, which adapts to the primal's decisions.

In addition to the adaptive behavior of the dual algorithm, the Lagrangian multipliers it picks control the range of the objective that the primal algorithm has to maximize.
For technical reasons (which we discuss in \cref{sec:prelim}), %unlike previous work (which we discuss in detail in \cref{sec:related}), 
we cannot a priori upper bound these multipliers. This means that the primal algorithm needs to maximize a function whose range is time-varying and unknown.
We develop algorithms that tackle this problem and offer regret bounds that would match the bounds we would get had the range been known in advance.
All our algorithms need to guarantee no-regret with time-varying ranges. We first solve this problem in \cref{ssec:full:good}.
In addition, our technique to handle this problem is very general and, we believe, is of independent interest.

\paragraph{From Standard Regret to Interval Regret} The `core' algorithm requires that the primal and dual algorithms have low interval regret (low regret in every interval of rounds).
While this is easily achieved by the Online Gradient Descent that is used for the dual algorithm, this is not automatically achieved by other algorithms, e.g., the Hedge algorithm \cite{DBLP:journals/jcss/FreundS97} has linear interval regret.
In \cref{ssec:full:interval}, we offer a black-box reduction to reduce the problem of standard regret minimization to interval regret minimization with only $\tilde O(\sqrt T)$ error, which also works for the above time-varying range problem.

Finally, we note that our resulting primal algorithm is a chain of algorithms that use the output of each other to generate their own output (we showcase this in \cref{fig:algo_struct}).
Instead of each algorithm outputting a bid that is used by the next algorithm, the output of each algorithm is a distribution over bids.
This makes the performance of an algorithm independent of the sampling of other algorithms; the only sampling performed is by the `final' algorithm.
This leads to simpler chaining of these algorithms and we believe would lead to much more stable guarantees in practice.

\boldparagraph{Bandit Information}
In \cref{sec:bandit} we consider the bandit information setting where the algorithm only observes whether the bid wins the auction or not and the price she pays if she wins.
In sharp contrast to the full-information setting, we prove an $\Omega(T^{2/3})$ lower bound on the regret for any online algorithm in first-price auctions.
While this is known for quasi-linear utility maximization \citet{DBLP:conf/nips/BalseiroGMMS19}, no results are known for value maximization.
Our lower bound is materialized in a very simple setting, as showcased in the theorem that follows.

\begin{theorem}[Informal version of \Cref{thm:bandit_lower_bound}]\label{thm:bandit_lb_intro}
    There exists an instance in value-maximizing first-price auctions where any algorithm with bandit information has regret $\Omega(T^{2/3})$, even when the bidder has a constant value, a total budget of $\Theta(T)$, and no ROI constraint.
\end{theorem}

Our lower bound is based on a distribution of highest competing bids that has the following property: for (almost) every pair of values in the support, there is an optimal bidding strategy that uses only those values.
We then slightly modify one value of that distribution adversarially such that our modification is (a) big enough to make every bid other than that value sub-optimal and (b) small enough that the bidder has to waste many rounds on sub-optimal bids before finding the optimal one.
This construction is inspired by the $\Omega(T^{2/3})$ lower bound of \cite{DBLP:conf/focs/KleinbergL03} who study revenue maximization for posted-price mechanisms without constraints.
%\gacomment{It would be good to mention the prior work that motivated this lower bound construction here}
%\gfcomment{Does this work?}
%\gacomment{yes}

%To complement our negative result shown in \Cref{thm:bandit_lb_intro}, we present an algorithm that achieves near-optimal regret bound when the bidder's value and the highest competing bid are independent and that the value distribution is known to the bidder.

%\begin{theorem}[Informal version of \Cref{thm:bandit_ub}]\label{thm:bandit_ub_intro}
%    Assume the bidder's value and the highest competing bid are drawn from two independent distributions $\calD_v$ and $\calD_d$. Then there exists an algorithm given $\calD_v$ and bandit information that achieves $\tilde O(T^{2/3})$ regret. 
%\end{theorem}

To complement our lower bound in \Cref{thm:bandit_lb_intro}, we present a $\tilde O(T^{3/4})$ regret upper bound.

\begin{theorem}[Informal version of \Cref{thm:bandit_ub}]\label{thm:bandit_ub_intro}
    There exists an algorithm given bandit information that achieves $\tilde O(T^{3/4})$ regret with respect to the best Lipschitz bidding function that always satisfies the ROI and budget constraints.
\end{theorem}

We remark that our regret bounds are based on similar results that approximately satisfy the ROI constraint (i.e. have sublinear violation with high probability).
In \Cref{sec:tight_roi}, we present a black box reduction to turn such an algorithm into an algorithm that strictly satisfies the ROI constraint.

Finally, we note that the focus of our full information results in \cref{sec:full} is regret minimization. {The algorithms we present make use of results in \citet{DBLP:journals/corr/CastiglioniCK23}, which requires exponential running time.} In \cref{sec:poly}, we present algorithms that require polynomial time to run and offer the same guarantees as \cref{thm:full_in_intro} when the values and highest competing bids are independent across rounds.

\subsection{Related work}\label{sec:related}

%{\bf Bidding in non-truthful auctions}: 
The most relevant paper to ours is \citet{DBLP:journals/corr/CastiglioniCK23}.
The algorithm designed in our paper is based on the primal/dual framework in \cite{DBLP:journals/corr/CastiglioniCK23}; we briefly introduce the framework in \cref{sec:prelim}.
They also use the framework to design algorithms for bidding in first-price auctions with budget and ROI constraints, albeit only for a finite number of values and bids: their regret bound is $\tilde O(\sqrt{n m T})$ against the best bid per value, where $n$ is the number of values and $m$ is the number of bids.
In addition, their algorithm satisfies the ROI constraint only approximately. In contrast, our results apply to continuous distributions and strictly satisfy the ROI constraint.

{\bf Online bidding in non-truthful auctions}: \citet{DBLP:journals/corr/LucierPSZ23} extend the results of \citet{DBLP:conf/innovations/GaitondeLLLS23}.
Their main result is an algorithm for bidding in first and second price auctions under budget and ROI constraints, that implies welfare guarantees when used by every player.
In addition, they prove that their algorithm, when used in a stochastic environment, has $\tilde O(T^{7/8})$ regret against the class of pacing multipliers while satisfying both constraints strictly.
\citet{DBLP:conf/sigecom/FikiorisT23} also focus on welfare guarantees in first-price auctions when budgeted players use arbitrary algorithms that have no-regret against the class of pacing multipliers.
In addition, they design a full information algorithm that has $\tilde O(\sqrt T)$ regret with respect to the same class in the stochastic environment.
{Finally, \citet{DBLP:conf/icml/WangYDK23} study online learning in first-price auctions with budgets but focus only on the independent values and highest competing bids.}

{\bf Online bidding in truthful auctions}: \citet{DBLP:conf/www/FengPW23} study online bidding in sequential truthful auctions under budget and ROI constraints in a stochastic environment.
Their algorithm guarantees $\tilde O(\sqrt T)$ regret with respect to the best bidding sequence and satisfies exactly both the ROI and budget constraints.
Their results are an extension of the results of \citet{DBLP:conf/sigecom/BalseiroG17,DBLP:conf/icml/BalseiroLM20} where they study the same setting without ROI constraints. On the other hand, our paper studies a more general class of (possibly) non-truthful auctions and provides an algorithm that has the same regret guarantee. \citet{DBLP:conf/sigecom/BalseiroG17,DBLP:conf/icml/BalseiroLM20} also study the adversarial setting, where the value and the highest competing bid are not sampled by a stationary distribution but are picked by an adversary.
In this setting, it is impossible to achieve regret that is sublinear in $T$, so they bound the competitive ratio, i.e. the multiplicative error, instead.
Aside from truthful auctions, \cite{DBLP:conf/icml/BalseiroLM20} also extend these guarantees to settings where the learner gets to observe every parameter of a round before picking a decision, e.g., in auctions observe the highest competing bid before bidding.

{\bf Online bidding without constraints}: Another line of work studies online bidding without constraints.
%\gadelete{~\citet{DBLP:journals/corr/HanZFOW20,DBLP:conf/colt/Cesa-BianchiGGG17}.} 
\cite{DBLP:journals/corr/HanZFOW20} study quasi-linear utility maximization in first-price auctions, while \cite{DBLP:conf/colt/Cesa-BianchiGGG17} study the maximization of arbitrary Lipschitz continuous functions. Both use an algorithm with a tree structure similar to ours; we comment on the similarity/differences in \Cref{ssec:full:tree}.
%\gadelete{
%There are other works that study online learning without constraints.}
\citet{DBLP:conf/nips/BalseiroGMMS19} study contextual online learning, which result into a $\tilde O(T^{2/3})$ regret bound for quasi-linear utility maximizers in first-price auctions with bandit feedback.
\citet{DBLP:conf/focs/KleinbergL03} study online pricing, where the learner wants to learn how to price an item to maximize revenue; one of their results implies that the above regret bound is tight.
% An extension of the previous result is in \citet{DBLP:conf/aistats/GolrezaeiJLM23}, who study pricing algorithms when the buyer is budget and ROI constrained and is using a no-regret algorithm.\gacomment{is the last one needed here?}

{\bf Online learning with constraints}: \textit{Bandits with Knapsacks} is a class of online learning problems where the learner has a general action space and multiple budget constraints~\cite{DBLP:conf/focs/BadanidiyuruKS13,DBLP:conf/sigecom/AgrawalD14,DBLP:conf/nips/AgrawalD16,DBLP:conf/focs/ImmorlicaSSS19,DBLP:conf/colt/Kesselheim020,DBLP:conf/colt/FikiorisT23}.
\citet{DBLP:conf/colt/SlivkinsSF23,DBLP:conf/nips/KumarK22,DBLP:journals/corr/BernasconiCCF23} study the same setting when the budget can also increase in some rounds.

{\bf Online learning without constraints}: Finally, the problem of online learning without constraints has received extensive attention.
\cite{DBLP:journals/ftopt/Hazan16} and \cite{DBLP:journals/sigecom/Slivkins20} are excellent textbooks.
The most commonly used algorithms for settings with finite number of actions are Hedge for full information feedback \cite{DBLP:journals/jcss/FreundS97} and EXP3 for bandit feedback \cite{DBLP:journals/siamcomp/AuerCFS02}.
For online convex optimization (where there are infinite number of actions) the most commonly used algorithm is Online Gradient Descent~\cite{DBLP:conf/icml/Zinkevich03}.

%% file: body/2.prelim.tex
\section{Preliminaries}\label{sec:prelim}

We consider the setting where a single bidder participates in $T$ sequential auctions.
In each round $t \in [T]$ there is a single item being sold; a pair $(v_t, d_t)$ is drawn i.i.d. from some unknown (joint) distribution $\calD$, where $v_t \in [0, 1]$ indicates the bidder's value at this round, and $d_t \in [0, 1]$ is the highest competing bid.\footnote{One can think of the highest competing bid as the highest bid among other bidders participating in the auction.}
The bidder submits a bid $b_t$ based on her value $v_t$ and the history. The bidder wins this round if her bid is at least\footnote{Our results easily extend to other tie breaking rules.} the highest competing bid $d_t$; we denote $x_t = \One{b_t \ge d_t}$. If the bidder wins the auction, then she pays a price $p_t = p(b_t, d_t)$, where $p(b,d) \in [0,1]$ is the payment function. For example, for first-price auctions, $p(b,d)=b$; for second-price auctions, $p(b,d)=d$; for any combination of the two auctions, $p(b,d)=q\cdot b + (1-q)\cdot d$ for some $q \in [0,1]$. We note that the payment function $p$ is fixed across all rounds and known to the bidder.

At the end of each round $t$, the bidder observes information about the current round depending on the feedback models we study. In the full-information setting, the bidder observes the highest competing bid $d_t$ with which she can compute the outcome for any possible bid at this round. In the bandit-information setting, the bidder only observes whether the bidder wins the auction or not (i.e. $x_t=\One{b_t \ge d_t}$) and the payment $p_t$ if she wins. 

Our results hold for different objectives of the bidder, who wants to maximize $\sum_{t\in [T]} u_t$, where $u_t$ is her per-round utility. 
The focus of our paper is value-maximizing, where $u_t = v_t x_t$, but our results also hold when the bidder has a quasi-linear utility, where $u_t = x_t(v_t - \nu p_t)$ for some $\nu \in [0, 1]$.

We assume that the bidder has a budget $B$. This is a strict upper bound on her total payment. Namely, it must hold that her total payment after $T$ rounds is at most $B$, i.e. $\sum_{t\in[T]}x_t p_t \le B$. We define $\rho = \frac{B}{T}$ and note that w.l.o.g. we can assume that $\rho \le 1$: any $\rho \ge 1$ implies that the bidder is effectively not budget constraint, since $p_t \le 1$. The bidder must also satisfy a Return-On-Investment (ROI) constraint.
Namely, her total value in the winning rounds must be at least a fraction of her total payment: $\sum_t x_t v_t \ge \g \cdot \sum_t x_t p_t$, for some $\g \ge 1$. 
For the ROI constraint, we often allow approximate satisfaction where $\sum_t x_t (v_t - \g p_t) \ge - V$ and $V$ is the violation amount; often $V = O(\sqrt T)$.
W.l.o.g., we assume that $\g = 1$; any other $\g$ can be handled by rescaling the values.\footnote{For any ROI $\g > 1$ we can rescale the values $v_t' := v_t / \g$.}

\boldparagraph{Benchmark}
We want the player to have low regret when competing against a class of bidding functions $\calF$ that map values to bids.
More specifically, we assume that for every $f \in \calF$, $f$ maps $[0, 1]$ to $[0, 1]$ and we want the player's resulting utility to be close to her utility if she knew the distribution $\calD$ in advance and she bid using the best fixed functions from $\calF$.
Since the bidder has to satisfy certain constraints, the best response to a distribution $\calD$ might be a distribution of functions over $\calF$, not a single function.
Our benchmark is the expected utility of the best distribution of functions from $\calF$ that satisfies the constraints in expectation.
For example, for a value maximizing player in first price auctions, i.e., $u_t = v_t \One{b_t \ge d_t}$ and $p_t = b_t$ we have
\begin{equation} \label{eq:pre:opt}
\begin{split}
    \opt = \;
    \sup_{F \in \Delta(\calF)} \quad&
    \E_{f\sim F}\Ex[(v, d)\sim\calD]{\big. v \cdot\One{f(v) \ge d}}
    \\
    \text{s.t.} \quad&
    \E_{f\sim F}\Ex[(v, d)\sim\calD]{\big. f(v) \cdot\One{f(v) \ge d}} \le \rho
    \\
    &
    \E_{f\sim F}\Ex[(v, d)\sim\calD]{\big. (v - f(v)) \cdot\One{f(v) \ge d}} \ge 0
\end{split}
\end{equation}
 
For simplicity, we assume that the function $f(v) = 0$ always belongs in $\calF$, making the above problem feasible for every such $\calF$ and $\calD$.
$T\cdot\opt$ is an upper bound for the achievable total expected utility of any algorithm that satisfies the constraints \cite{DBLP:conf/focs/BadanidiyuruKS13,DBLP:conf/icml/CastiglioniCK22,DBLP:journals/corr/CastiglioniCK23}.
We design algorithms that have low regret with respect to $T \cdot \opt$ (where $\opt$ is defined analogously depending on the buyer's objective and auction format).

\boldparagraph{Primal/dual framework}
We now briefly describe the framework developed by \cite{DBLP:journals/corr/CastiglioniCK23} for online learning under budget and ROI constraints, where the constrained problem is solved by having two learning algorithms, the primal and the dual, compete against each other in a sequential unconstrained zero-sum game.
We first discuss the assumptions required on the input distribution and how these relate to our setting of learning in sequential auctions.
We then present the guarantees that the primal algorithm needs to satisfy to get guarantees for the original constrained problem.
For simplicity, the rest of this section focuses on value maximization and presents all the assumptions in the context of auctions.
We refer the reader to their paper for a more comprehensive description of their techniques.

First, we illustrate the need of some assumptions on the distribution $\calD$ that generates $v_t, d_t$.
Consider there is no budget constraint and the auction is first-price.
$\calD$ is such that every round $v_t = 1/2$, while $d_t = 1/2$ with probability $1/2$ and $d_t \sim \text{Uniform}[1/2,1]$ otherwise.
Under this $\calD$, bidding $b_t = 1/2$ is optimal since lower bids always lose the auction and higher bids could imply a violation of the ROI constraint that cannot be reversed\footnote{We note that in value maximization, bidding higher than one's value can be part of the optimal strategy.}.
However, without knowledge of $\calD$, an online learner might make a `wrong bid' whose results are irreversible.
In other words, $\calD$ should be so that the bidder should be able to reverse wrong decisions made in the early rounds.

\cite{DBLP:journals/corr/CastiglioniCK23} solve the above problem by assuming some structure on $\calD$ and $\calF$.
We present their assumption in \cref{sec:app:prelim}, but in short, it requires that the constraints in \eqref{eq:pre:opt} can be satisfied with some slack.
Instead, we focus on an assumption that is more meaningful for the setting of auctions and implies the aforementioned one.
Specifically, we assume that there exists a bidding function $f \in \calF$ that on expectation leads to $\b$ more value than payment for some $\b \ge 0$:
\begin{equation}\label{eq:pre:beta}
    \exists f \in \calF: \quad
    \Ex[(v,d) \sim \calD]{
        \qty\Big( v - p\qty\big(f(v), d) ) \One{f(v) \ge d}
    }
    \ge
    \b
    .
\end{equation}

This is similar to the assumption \cite{DBLP:conf/www/FengPW23} make, who examine value maximization in repeated truthful auctions.
Their assumption is the same as \eqref{eq:pre:beta} but for $f(v) = v$.
This assumption might seem stronger, since it implies \eqref{eq:pre:beta} when the function $f(v) = v$ is contained in $\calF$.
However, the reverse also holds for truthful auctions, since $f(v) = v$ is the optimal bidding function for maximizing quasi-linear utility.

We now show the basics of the primal/dual framework and the guarantees the primal needs to satisfy to get guarantees for the original problem.
We first define the following function:
\begin{equation*}
    \calL_t(b, \l, \mu)
    = 
    \One{b \ge d_t} \qty\big(
        v_t - \l b + \mu v_t - \mu b
    )
    +
    \l \rho 
    .
\end{equation*}
This function is inspired by the optimization problem in \eqref{eq:pre:opt} (and would analogously be defined for other objectives and pricing functions).
$\l, \mu$ are Lagrange multipliers that correspond to the budget and ROI constraint, respectively.
The core algorithm of \cite{DBLP:journals/corr/CastiglioniCK23} runs two algorithms, the primal algorithm that picks bids and the dual that picks Lagrange multipliers.
On every round $t$, the primal (resp. dual) algorithm picks $b_t$ (resp. $(\l_t, \mu_t)$) aiming to maximize (resp. minimize) $\calL_t(b_t, \l_t, \mu_t)$.
Given their actions, each algorithm faces some regret.
To get regret guarantees for the original problem (value maximization under constraints) the primal and dual algorithms must have with high probability low \textit{interval regret}, i.e., have low regret over every interval $[\tau_1, \tau_2] \sub [T]$.
Before formally defining this, we point out one subtle detail.

In general, regret bounds depend on the range of values that the objective function takes.
This range for the primal algorithm in our setting is $\sup_b \calL_t(b, \l_t, \mu_t) - \inf_b \calL_t(b, \l_t, \mu_t)$ in round $t$, which can be as high as $2 \mu_t + \l_t + 1$.
The last expression depends on $\l_t, \mu_t$ which are picked by the dual algorithm, meaning that we cannot know $\max_t\{2 \mu_t + \l_t + 1\}$ in advance.
One solution to this is to explicitly upper bound $\l_t$ and $\mu_t$.
To get meaningful regret guarantees such an upper bound would be $\l_t,\mu_t \le \frac{1}{\b\rho}$.
However, calculating $\b$ requires $\calD$, as seen in \eqref{eq:pre:beta}, which is unknown.

\cite{DBLP:conf/icml/CastiglioniCK22} show that the above issue can be circumvented, if the primal algorithm satisfies stronger \textit{interval regret} bounds: with high probability, for every interval of rounds, the regret in those rounds is small with respect to the best fixed action in that interval and the maximum Lagrange multipliers seen so far.
To formally define this, first define $M_t = \max_{t' \le t}\{2 \mu_{t'} + \l_{t'} + 1\}$.
We require that the primal algorithm picks bids $b_1, \ldots, b_T$ such that for any $\l_1,\mu_1,\ldots,\l_T,\mu_T$ that could be picked adaptively, there exists a function $\regP$ such that for every $\d \in (0, 1]$,
\begin{equation}\label{eq:pre:primal_regret}
    \Pr{
        % \forall 1 \le \tau_1 < \tau_2 \le T: \;
        \forall [\tau_1, \tau_2] \sub [T]: \;
        \sup_{f\in\calF} \sum_{t=\tau_1}^{\tau_2} \calL_t\qty\big(f(v_t), \l_t, \mu_t)
        -
        \sum_{t=\tau_1}^{\tau_2} \calL_t(b_t, \l_t, \mu_t)
        \le
        \regP(T, \d, M_{\tau_2})
    }
    \ge 1 - \d
\end{equation}
% % 
% and
% \begin{equation}\label{eq:pre:dual_regret}
%     \Pr{
%         \forall 1 \le \tau_1 < \tau_2 \le T: \;
%         \sum_{t=\tau_1}^{\tau_2} \calL_t(b_t, \l_t, \mu_t)
%         -
%         \inf_{\l\ge 0,\mu\ge 0} \sum_{t=\tau_1}^{\tau_2} \calL_t(b_t, \l, \mu)
%         \le
%         \regD(T, \d, U_{\tau_2})
%     }
%     \ge 1 - \d
% \end{equation}

% A similar condition is satisfied for the dual algorithm using Online Gradient Descent \cite{DBLP:conf/icml/Zinkevich03}.

Given the above guarantee for the primal and some mild conditions on $\regP$, we get the guarantees for the original problem.
More specifically, we require that the dependence of $M$ in $\regP(T, \d, M)$ is not worse than quadratic.
Using this condition, both the regret of the original problem and the violation of the ROI constraint are at most $\regP(T,\d,\frac{2}{\b\rho})$.
Formally, we have the following theorem.

% Thus, in order to get meaningful regret bounds we need a bound on $U_\tau$.
% \cite{DBLP:journals/corr/CastiglioniCK23} show the following: if \eqref{eq:pre:primal_regret} holds with $\regP(T, \d, U_{\tau}) \le U_\tau^2 \order{\regP(T, \d, 1)}$, then it must also hold that $U_T \le \frac{2}{\a}$.

\begin{theorem}[Theorem 6.9 of \cite{DBLP:journals/corr/CastiglioniCK23}, adapted to auctions]\label{thm:pre:primal_dual}
    Assume \eqref{eq:pre:beta} is true for $\b > 0$, the dual algorithm is Online Gradient Descent, and \eqref{eq:pre:primal_regret} holds with $\regP(T, \d, M) \le \order{M^2 \regP(T, \d, 1)}$.
    Then, for every $\d > 0$, with probability at least $1 - 4\d$,
    \begin{equation}\label{eq:pre:final_regret}
        T \cdot \opt
        -
        \sum_{t\in [T]} v_t \One{b_t' \ge d_t}
        \le
        \order{
            \regP\qty( T, \d, \frac{2}{\b \rho} ) + \frac{1}{\b\rho}\sqrt{T \log (T/\d)}
        }
    \end{equation}
    and
    \begin{equation*}
        \sum_{t\in [T]} \qty\big( v_t - p(b_t', d_t) ) \One{b_t' \ge d_t}
        \ge
        - \order{
            \regP\qty( T, \d, \frac{2}{\b \rho} ) + \frac{1}{\b\rho}\sqrt{T \log (T/\d)}
        }
        ,
    \end{equation*}
    where $b_t'$ is the primal's bid in round $t$ unless the remaining budget is less than $1$; otherwise $b_t' = 0$.
\end{theorem}

% One assumption that is not explicitly stated in the theorem is that neither the primal nor the dual algorithms need to know the value of $\b$ or $\a$.

We note that if $|\calF| = K$ was finite, we had bandit feedback, and an upper bound $\bar M$ for $M_T$ then existing techniques would allow us to get $\regP(T, \d, M_t) = \tilde O(\bar M\sqrt {T K \log(1/\d)})$.
In the setting where there is no such upper bound, \cite{DBLP:journals/corr/CastiglioniCK23} offer an algorithm with $\regP(T, \d, M_t) = \tilde O(M_t^2\sqrt {T K \log(1/\d)})$.
While this regret bound satisfies the requirements of \cref{thm:pre:primal_dual}, the quadratic dependence on $M_t$ is sub-optimal, making the resulting regret bound in \eqref{eq:pre:final_regret} be proportional to $\frac{1}{\b^2\rho^2}$.
One of the contributions of the following sections is a general technique to get linear dependence of $M_t$ in $\regP(T, \d, M_t)$, similar to knowing $M_t$ in advance.
This leads to much improved regret bounds, especially when $\b$ and $\rho$ are small.

%% file: body/3.full.tex
\section{Primal algorithm designs with full information} \label{sec:full}

In this section, we design a primal algorithm that satisfies \eqref{eq:pre:primal_regret} for sequential auctions.
Our goal is to pick bids to maximize the Lagrangian $\calL_t(b_t, \l_t, \mu_t)$.
To simplify our presentation, for every round $t$, we define $r_t(\cdot)$ to be the part of the Lagrangian that depends on the bid:
\begin{equation}\label{eq:full:reward}
    r_t(b)
    =
    \One{b \ge d_t} \qty\big(
        \chi_t v_t - \psi_t p(b, d_t)
    )
\end{equation}
where $d_t \in [0, 1]$ is the highest competing bid, $v_t \in [0, 1]$ is the bidder's, $p(b, d_t)$ is the payment if the bidder bids $b$, and $\chi_t, \psi_t$ are non-negative numbers that depend on the Lagrange multipliers of round $t$.
For a value maximizing we have $\chi_t = 1 + \mu_t$ and $\psi_t = \l_t + \mu_t$.
For quasi-linear utility $\One{b_t \ge d_t}(v_t - \nu p(b_t, d_t))$ for some $\nu \in [0, 1]$ we have $\chi_t = 1 + \mu_t$ and $\psi_t = \nu + \l_t + \mu_t$.
To present the general results, we assume that $\chi_t, \psi_t$ are arbitrarily picked by an adaptive adversary such that $\chi_t > 0$ and $\psi_t \ge 0$.
We emphasize the lack of an upper bound on $\chi_t, \psi_t$ and recall \eqref{eq:pre:primal_regret}, which requires that the regret scales with the largest $\chi_t, \psi_t$ seen so far.
We overload the notation of $r_t(\cdot)$ to also take as an argument a bidding function $f : [0,1] \to [0,1]$, in which case $r_t(f) = r_t(f(v_t))$.

Our results apply to a wide range of price functions $p(\cdot,\cdot)$ that include any combination of first and second price.
In particular, $p(\cdot,\cdot)$ needs to satisfy the following assumption which we explain after its formal statement.
% We note (and prove afterward) that the following assumption is satisfied by any mixture of first and second price.

\begin{assumption}\label{asmp:full}
    The function $p: [0, 1]^2 \to [0, 1]$ satisfies: (i) $p(0, 0) = 0$; (ii) for all $d$, $p(b, d)$ is non-decreasing and $1$-Lipschitz continuous\footnote{A function $g:\R \to \R$ is $L$-Lipschitz continuous if $|g(x) - g(y)| \le L |x - y|$ for all $x,y\in\R$.} in $b$; (iii) for every $t$, there exists a ``safe'' bid $b_t^\circ$ so that 
    \begin{enumerate}[label=(\alph*),leftmargin=20pt]
        \item\label{cond:1} $r_t(b_t^\circ) \ge 0$ for all $d_t$.
        
        \item\label{cond:2} $b_t^\circ$ is a function of $v_t, \chi_t, \psi_t$ but not $d_t$.
        
        \item\label{cond:3} if for some bid $b$ it holds $\inf_{d_t \in [0,1]} r_t(b) < 0$, then for all $d_t$ it holds $r_t(b_t^\circ) \ge r_t(b)$.
    \end{enumerate}
\end{assumption}

%While the conditions before the safe bid might seem reasonable, the latter is both hard to understand and its usefulness is unclear.
While Conditions (i) and (ii) are standard assumptions on the payment function, Condition (iii) is less straightforward, which we explain below. In short, the safe bid is a way to guarantee that our algorithm will never use $b_t$ with $r_t(b_t) < 0$.
%; we discuss this in length after we describe \ref{cond:1}, \ref{cond:2}, \ref{cond:3}.
For every $t$, we require that the function $r_t(b) = \One{b \ge d_t} ( \chi_t v_t - \psi_t p(b, d_t) )$ has a ``safe'' bid $b_t^\circ$ such that:
\ref{cond:1} $b_t^\circ$ guarantees a non-negative reward,
\ref{cond:2} $b_t^\circ$ can be calculated using the information that the bidder has before submitting bids, i.e., $v_t, \chi_t, \psi_t$, and
\ref{cond:3} $b_t^\circ$ guarantees reward that is at least the reward of any other bid which has a negative reward for some $d_t$.
We emphasize that \ref{cond:3} states that $r_t(b_t^\circ) \ge r_t(b)$ for all $d_t$ as long as $r_t(b) < 0$ for some $d_t$.

We now show that any mixture of first and second price auction satisfies \cref{asmp:full} and more specifically Condition (iii).
\begin{itemize}
    \item For first-price, i.e., $p(b, d) = b$, a safe bid is $b_t^\circ = 0$.
    Conditions \ref{cond:1} and \ref{cond:2} are obvious.
    For Condition \ref{cond:3} we notice that $\inf_{d_t} r_t(b) < 0 \iff \psi_t b > \chi_t v_t$.
    For such bids $b$, $r_t(b) \le 0$ for all $d_t$ which makes Condition \ref{cond:3} follow from $r_t(0) \ge 0$.
    
    \item For second-price, i.e., $p(b, d) = d$, a safe bid is $b_t^\circ = \min\big\{ \frac{\chi_t}{\psi_t} v_t , 1 \big\}$.
    Condition \ref{cond:2} follows from definition and Condition \ref{cond:1} follows by calculating $r_t(b_t^\circ) = (\chi_t v_t - \psi_t d_t)^+$.
    For condition \ref{cond:3}, we notice that $\inf_{d_t} r_t(b) < 0 \iff b > \frac{\chi_t}{\psi_t} v_t$; this is because $r_t(b) < 0 \iff b \ge d_t > \frac{\chi_t}{\psi_t} v_t$ and the last inequality can be satisfied iff $b > \frac{\chi_t}{\psi_t} v_t$.
    For such $b$, $r_t(b_t^\circ) < r_t(b)$ would hold only if the bid $b_t^\circ$ could win an auction that $b$ would not, which is impossible since $b > b_t^\circ$; this implies \ref{cond:3}.
    
    We note that the above safe bid $b_t^\circ$ satisfies something much stronger than Condition \ref{cond:3}, as we mentioned in the introduction.
    For every $d_t$, it holds that $b_t^\circ \in \argmax_b r_t(b)$.
    
    \item For a mix of first and second price, i.e., $p(b, d) = q b + (1-q) d$ for some $q \in [0, 1]$, the safe bid for second-price also works here: $b_t^\circ = \min\big\{ \frac{\chi_t}{\psi_t} v_t , 1 \big\}$.
    In this case $r_t(b_t^\circ) = (1-q)(\chi_t v_t -  \psi_t d_t)^+$, which proves Condition \ref{cond:2}.
    Similar to before, $\inf_{d_t} r_t(b) < 0 \iff b > \frac{\chi_t}{\psi_t} v_t$ in which case an analysis like the ones above proves Condition \ref{cond:3}.
\end{itemize}

\subsection{Overview of the Primal Algorithm Design}\label{subsec:full_info_overview}

We now state the main result of the section: There is an algorithm that has low interval regret with high probability.
Our regret guarantee is with respect to the class of all $L$-Lipschitz continuous functions for any $L \ge 1$, which we denote with $\calF$.
In addition, our regret bound in rounds up to $t$ scale linearly with respect to the highest $\chi, \psi$ seen so far: $U_\tau = \max_{t\le \tau}\{\chi_t, \psi_t\}$.
Finally, we note that our algorithm takes into advantage the fact that $\chi_t,\psi_t$ are known before bidding in round $t$.
This is important for various calculations, e.g., calculating the safe bid $b_t^\circ$ of a round.
Outside of that, however, $\chi_t,\psi_t$ are picked adversarially and are not known before round $t$.

\begin{theorem} \label{thm:full:primal}
    Let $\calF$ be the set of all $L$-Lipschitz continuous functions from $[0,1]$ to $[0,1]$ for some $L \ge 1$.
    Assume that the payment function satisfies \cref{asmp:full}.
    Assume that $\chi_t,\psi_t$ are picked by an adaptive adversary in round $t$ and are revealed to the algorithm after round $t-1$.
    Then there exists an algorithm that generates bids $b_1, \ldots, b_T$ such that for all $\d > 0$, with probability at least $1 - \d$ it holds
    \begin{equation*}
        \forall [\tau_1, \tau_2] \sub [T] : \quad
        \sup_{f \in \calF} \sum_{t = \tau_1}^{\tau_2} r_t(f)
        -
        \sum_{t = \tau_1}^{\tau_2} r_t(b_t)
        \le
        \order{
            U_{\tau_2} \sqrt{L T} \log T
            +
            % \sqrt{T \log T}
            % +
            U_{\tau_2} \sqrt{T \log(T/\d)}
        }
    \end{equation*}
    where $U_{\tau_2} = \max_{t \le \tau_2} \{\chi_t, \psi_t\}$.
\end{theorem}

\cref{thm:full:primal} satisfies \eqref{eq:pre:primal_regret} as well as the conditions of \cref{thm:pre:primal_dual} getting the following theorem.

\begin{theorem}\label{thm:full:main}
    There is an algorithm for value or quasi-linear utility maximization when the payment function satisfies \cref{asmp:full}, such that for every $\d > 0$, with probability at least $1 - \d$ the algorithm has regret against the class of $L$-Lipschitz continuous functions and ROI violation at most
    $$
        \frac{1}{\rho \b} \, \order*{
            \sqrt{L T} \log T
            +
            \sqrt{T \log(1/\d)}
        }
    .$$
\end{theorem}

There are a couple of technical challenges in order to get low regret against $\calF$ and prove the first theorem, \cref{thm:full:primal}. First, even for in-expectation regret bound, we cannot directly use a standard algorithm like Hedge, since $\calF$ contains an infinite number of actions.
Instead, we work with finite sets that approximate $\calF$.
More specifically, for an accuracy parameter $\e > 0$ we define $\calF_\e \sub \calF$ such that
\begin{equation}\label{eq:full:lip}
    \forall f \in \calF,
    \exists f^\e \in \calF^\e:\quad
    f(v)
    \le
    f^\e(v)
    \le
    f(v) + \e
    ,\;
    \forall v \in[0, 1]
\end{equation}
i.e., for every $f$ in the original set $\calF$ there exists some function $f^\e$ in the new set $\calF_\e$ such that $f^\e$ is at least $f$, but is never greater by more than $\e$.
Previous work shows that there is always an $\calF_\e$ with $|\calF_\e| \le \exp(O(L / \e))$. \cite[Corollary 2.7.2.]{vaart1997weak}

The above bound on the cardinality of $\calF_\e$ is exponential in $1/\e$.
This means that even on this discrete set of actions, using Hedge does not lead to optimal regret bounds.
Since for $K$ actions the regret of Hedge is $O(\sqrt{T \log K})$ and if we assume that the per round error of using $\calF_\e$ instead of $\calF$ is $O(\e)$, we get an overall regret of $O(T \e + \sqrt {T / \e})$ (assuming $L = 1$).
Picking $\e = \Theta(T^{-1/3})$ leads $O(T^{2/3})$ regret.
However, as \cref{thm:full:primal} suggests, we can get much stronger $\tilde O(\sqrt T)$ regret bounds.

Before we discuss how we get the improved regret, we point out another technical challenge regarding the discretization error: In the above simplistic analysis, we assumed that the error of using $\calF_\e$ instead of $\calF$ is $O(\e)$ in every round.
However, this is not the case:
Fix $f,f^\e$ as described in \eqref{eq:full:lip}.
Using $f^\e$ instead of $f$ in a round $t$ results in at most $\psi_t \e$ error if bidding $f(v_t)$ wins the auction (since $f^\e(v_t)$ also wins).
However, if $f(v_t) < d_t \le f^\e(v_t)$ it might be the case that $r_t(f) = 0 \gg r_t(f^\e)$.
For example, in first-price auctions, if $\chi_t = \psi_t = 1$, $v_t = \e$, $f(v_t) = 1 - \e$, $f^\e(v_t) = 1$, and $d_t = 1 - \e/2$, we have $r_t(f^\e) = - (1 - \e) \ll r_t(f) = 0$.

The safe bids of \cref{asmp:full} solve this discretization issue.
% The above discussion illustrates how we can use keep the discretization error to $O(\e)$ per round.
If a bid $b + \e$ is in danger of leading to a negative reward (i.e., $\inf_{d_t} r_t(b + \e) < 0$), we can use the safe bid instead to guarantee an at least as good reward for any $d_t$, which in turn is $O(\e)$ worse than $r_t(b)$.
We note that one can circumvent this issue when maximizing quasi-linear utility with no constraints (in which case $\chi_t = \psi_t = 1$) by making sure that the class $\calF$ contains only functions such that $f(v) \le v$, which is the solution of \cite{DBLP:journals/corr/HanZFOW20}.
However, we cannot limit $\calF$ in such a way since $\chi_t, \psi_t$ change dynamically.

Having explained how we solve the discretization error of using $\calF_\e$ instead of $\calF$, we switch our attention back to the issue that $\calF_\e$ is exponential in $1/\e$. To get the optimal $\tilde O(\sqrt T)$ regret bound instead of $\tilde O(T^{2/3})$, we use an algorithm that utilizes the structure of the functions in $\calF_\e$.
More specifically, we use the fact that for every $\e' \ge \e$ we can partition the functions of $\calF_\e$ into $\exp(O(L/\e'))$ disjoint groups where the functions of each group are $\e'$-close.
We thus create a sequence of such groups for $\e' = 2\e, 4\e, 8\e, \ldots$, where each such group is exponentially smaller than the previous one.
We then create a hierarchical tree structure out of these groups, where the leaves of the tree are the functions of $\calF_\e$ and a non-leaf node $i$ levels above the leaves represents a group for $\e' = 2^i\e$.
Leaves whose distance is larger represent functions that are `far apart' while smaller distances indicate that the functions are `close'.
Using this structure, each node in the tree suggests a bidding function; leaves suggest the bidding function they represent (replacing it with the safe bid if it might lead to negative reward) and non-leaf nodes use a no-regret algorithm to select functions from their children.
Non-leaf nodes utilize the proximity of their children's functions to get enhanced regret bounds, which is the key to the improved regret.
Finally, the whole algorithm's output is the root's bid.
We develop this tree algorithm in \cref{ssec:full:tree}.

In \cref{ssec:full:good}, we develop the algorithm that is used by the non-leaf nodes and that utilizes the proximity of its children's bids.
We present this algorithm in the general language of online regret maximization.
The structure that this algorithm takes advantage of is that there is a ``good'' action: it is at most $\Delta$ sub-optimal compared to any other action, in every round.
% The existence of the good action is what allows the non-leaf nodes of the tree algorithm to get better guarantees than just running Hedge over the functions of $\calF_\e$.
This leads to $\order*{\sqrt{\Delta T \log K}}$ regret, which offers a great improvement over the regret of $\order*{\sqrt{T \log K}}$ that Hedge has when $\Delta \ll 1$.
This regret matches the one in \cite{DBLP:journals/corr/HanZFOW20}.
In addition, our algorithm works when the range of the rewards is time-varying and unknown in advance; we also offer a matching high probability regret bound.

Finally, to complete the proof of \cref{thm:full:primal}, in \cref{ssec:full:interval}, we reduce the problem of bounding interval regret to bounding normal regret.
The whole structure of our algorithm (including the primal/dual framework) is shown in \cref{fig:algo_struct}.

\begin{figure}[t]
    \centering
    \includegraphics[height=.16\textheight]{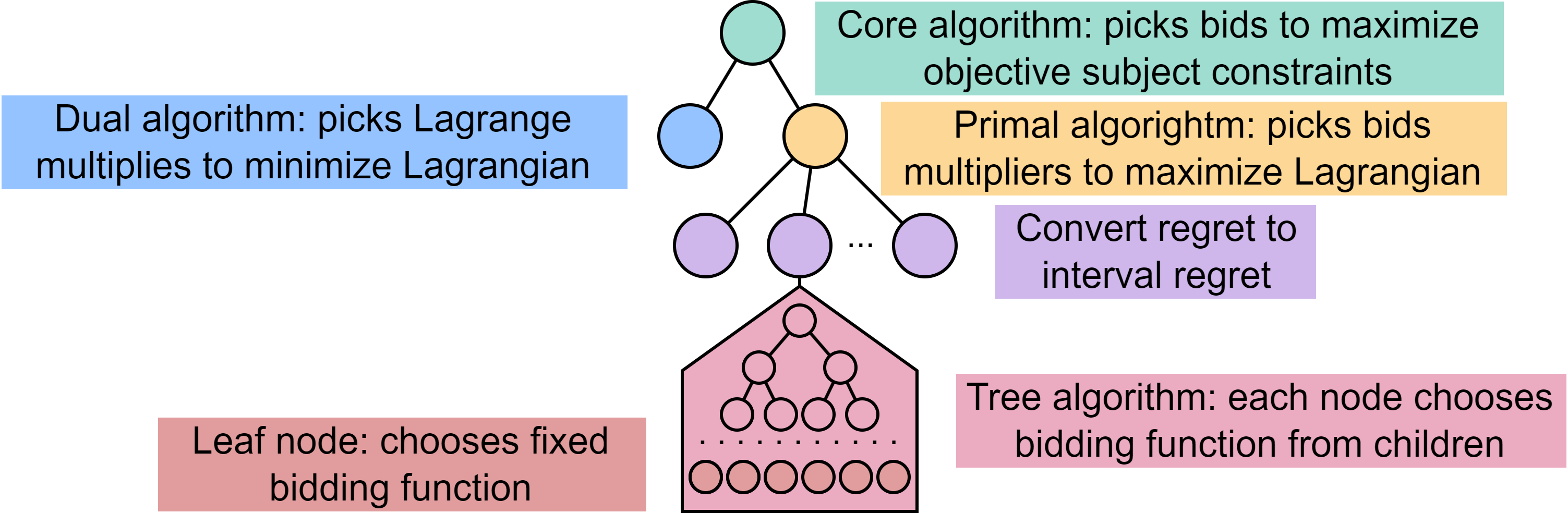}
    \caption{The algorithm structure of the entire primal/dual framework in our setting.\vspace{-7pt}}
    \label{fig:algo_struct}
\end{figure}

% We make one final note.
Our resulting primal algorithm works by chaining multiple sub-algorithms, as shown in \cref{fig:algo_struct}.
The standard way to employ this chaining is to use each algorithm's outputting action.
However, since our goal is a regret bound with high probability, the success of every algorithm is conditioned on the success of all its sub-algorithms. This creates noise in the high probability bounds, given the fact that the number of sub-algorithms for a non-leaf node can be exponential.

To overcome this challenge, we employ a different technique.
Instead of an action, each algorithm outputs the distribution from which it would have sampled its action.
These distributions satisfy a regret bound with probability $1$, making chaining multiple algorithms much more stable since the sampling of an action happens only once and not for every sub-algorithm.

\subsection{Time-varying Reward Ranges and Good Actions}\label{ssec:full:good}

In this section we develop the algorithm that we need for the non-leaf nodes of the tree algorithm, which takes advantage of the proximity of the rewards of its actions.
We present the algorithm in the general language of online learning when there is a set $[K]$ of actions ($K \ge 2$) and an arbitrary reward function $r_t(\cdot)$ for every round $t$ which the learner observes after round $t$.

There has been extensive work on this setting, under various different assumptions that make the problem easier or harder.
Our contribution to this space is twofold.
First, our algorithm does not need to know the range of the rewards in future rounds.
More specifically, we assume that the reward of round $t$ is in the range $[0, U_t]$ for some adversarially chosen $U_t > 0$ and the learner observes $U_t$ when picking her action $a_t$ in round $t$.
We assume that $U_1 \le U_2 \le \ldots \le U_T = U$.
If $U$ is known in advance, then using Hedge leads to regret that scales linearly with $U$.
We achieve the same dependency without knowing $U$.
This also improves the quadratic dependency on $U$ of previous work \cite[Theorem 8.1]{DBLP:journals/corr/CastiglioniCK23}.

Our second contribution is to apply this setting to the setting where there is a $\Delta$-good action.
The formal definition is in \cref{def:full:good}, but simply put a $\Delta$-good action is at most $\Delta \cdot U_t$ sub-optimal compared to any other action in every round $t$.
\cref{def:full:good} extends the original definition of \cite{DBLP:journals/corr/HanZFOW20} for time-varying reward ranges $U_1, \ldots, U_T$.

\begin{definition}[Good action]\label{def:full:good}
    For rewards $r_t: [K] \to [0, U_t]$, a \textit{$\Delta$-good action} $g \in [K]$ is such that
    \begin{equation*}
        \forall t\in [T], a\in [K]:\quad
        r_t(g) \ge r_t(a) - \Delta U_t
    \end{equation*}
\end{definition}

\noindent
Taking $U_t = 1$ yields the definition of \cite{DBLP:journals/corr/HanZFOW20}.
Note that $\Delta \in [0,1]$ since $0 \le r_t(a) \le U_t$.

We now present the algorithm for the above setting.
Our algorithm is the Hedge algorithm but with a carefully selected step size $\eta_t$.
In particular, the probability of playing action $a$ in round $t$ is proportional to $\exp(\eta_t R_{t-1}(a))$, where $R_{t-1}(a) = \sum_{t'\le t-1} r_{t'}(a)$ and $\eta_t \propto 1/U_t$.
We believe that this step size is of independent interest and can be used in any online learning setting to get regret bounds for time-varying ranges that match classical ones.
Our regret bound is $\order*{U \sqrt{T \Delta \log K}}$, which matches the one in \cite{DBLP:journals/corr/HanZFOW20}, without assuming $U_t = U$ for every $t$ and that $U$ is known.
We note that \cite{DBLP:journals/corr/HanZFOW20} prove that this regret bound is optimal for $\Delta \ge \Omega(\log^2T / T)$.
The full algorithm can be found in \cref{algo:full:good}.

\begin{algorithm}[t]
% \SetAlgoNoLine
\DontPrintSemicolon
\caption{Hedge for time-varying ranges and good actions}
\label{algo:full:good}
\KwIn{Number of rounds $T$, number of actions $K$, sub-optimality of good action $\Delta$} 

Initialize cumulative reward of each action $R_0(a) = 0 \quad \forall \, a\in [K]$\;
\For{$t \in [T]$}
{
    Receive reward range $[0, U_t]$ and calculate step size $\eta_t = \frac{1}{U_t} \sqrt{\frac{\log K}{T \Delta}}$\;
    
    Calculate probability distribution $p_t(a) = \exp\qty\big( \eta_t R_{t-1}(a) ) / \sum_{a'}\exp\qty\big( \eta_t R_{t-1}(a') ) \quad \forall \, a \in [K]$\;

    Sample and play action $a_t \sim p_t(a)$\;
    
    Receive rewards $r_t : [K] \to [0, U_t]$\;
    
    Update the cumulative rewards $R_t(a) = r_t(a) + R_{t-1}(a) \quad \forall \, a\in [K]$\;
}
\end{algorithm}

We now present our regret bound.
As mentioned before, we bound the regret of the action distributions $p_1(\cdot), \ldots, p_T(\cdot)$ that \cref{algo:full:good} generates instead of bounding the regret of sampled the actions $a_1, \ldots, a_T$.
Our bound can be used to get a matching bound on the expected regret. Then using standard concentration inequalities we can get bounds with high probability.
However, the most useful application of outputting distributions instead of actions is that chaining multiple no-regret algorithms becomes easier.
Instead of each algorithm suggesting an action, it suggests a distribution of actions.
This allows for stronger high-probability guarantees: the overall regret bound is dependent only on one sampling process, the one performed by the top-level algorithm.

\begin{theorem}\label{thm:full:good1}
    Assume that in round $t$ the reward function $r_t:[K] \to [0, U_t]$ is picked by an adaptive adversary, where $U_1 \le \ldots \le U_T$.
    Assume that there is a $\Delta$-good action.
    If $\Delta \ge \frac{4\log K}{T}$ then the action distributions $p_1(\cdot), \ldots, p_T(\cdot)$ generated by \cref{algo:full:good} guarantee with probability $1$:
    \begin{equation*}
        \forall t \in [T] : \quad
        \max_{a\in[K]} \sum_{t \in [\tau]} r_t(a)
        -
        \sum_{t \in [\tau]} \sum_{a\in[K]} p_t(a) r_t(a)
        \le
        4 U_\tau \sqrt{T \Delta \log K}
        .
    \end{equation*}
\end{theorem}

\cref{algo:full:good} dynamically adapts to the time varying reward ranges due to $\eta_t \propto 1 / U_t$.
If the adversary picks $U_t \gg U_{t-1}$ then because of this definition of $\eta_t$, our algorithm adapts to the new ranges and picks the action of round $t$ accordingly.
We believe this technique is very versatile and can be used to modify existing no-regret algorithms to work in the space of time-varying reward ranges.
We show this in the bandit information setting in \cref{thm:bandit:gen_algo}.
We defer the full proof of \cref{thm:full:good1}, along with the proofs of the other results of this section to \cref{sec:app:full}.
In \cref{sec:app:full} we also include a high probability version of the above theorem, \cref{thm:full:good2} with $\order*{U \sqrt{T\Delta\log(1/\d)}}$ additional error.
To retain regret that depends on $T\Delta$ instead of $T$, we cannot directly apply a standard concentration inequality and instead show that our rewards have low variance that depends on $\Delta$ to get the improved bound.

\subsection{Tree Algorithm for Learning to Bid with Lipschitz Functions}\label{ssec:full:tree}

In this section we describe the tree algorithm that maximizes $r_t(\cdot)$ (as defined in \eqref{eq:full:reward}) and has low regret compared to the best $L$-Lipschitz function in the class $\calF$.
Our algorithm shares multiple similarities with the algorithms of \cite{DBLP:journals/corr/HanZFOW20,DBLP:conf/colt/Cesa-BianchiGGG17}.
However, due to $\chi_t,\psi_t$ being arbitrary, we need to make non-trivial adjustments to get the desired guarantees.

As we mention at the beginning of \cref{sec:full}, we focus on subsets of $\calF$ that approximate $\calF$ and have finite cardinality.
We use a slightly different notation and define the sets $\calF_0,\calF_1,\ldots,\calF_M$ (for some $M \in \N$) such that for every $i = 0, \ldots, M$ it holds that
\begin{equation} \label{eq:full:cover}
    \forall f \in \calF,
    \exists f' \in \calF_i:\quad
    f(v)
    \le
    f'(v)
    \le
    f(v) + 2^{-i}
    ,\;
    \forall v \in[0, 1]
    .
\end{equation}

Our algorithm maintains a distribution over the functions in $\calF_M$ and bids according to it.
However, instead of directly using a function $f_M \in \calF_M$ and its bid $f_M(v_t)$ in round $t$, we modify that bid, since for certain values of $d_t$, $r_t(f_M)$ might be negative.
This is where we use \cref{asmp:full} and define
\begin{equation} \label{eq:full:good}
    b_t^{f_M}
    =
    \begin{cases}
        f_M(v_t), &\text{ if } \inf_{d_t} r_t\big( f_M(v_t) \big) \ge 0 \\
        b_t^\circ, &\text{ if } \inf_{d_t} r_t\big( f_M(v_t) \big) < 0
    \end{cases}
\end{equation}
where $b_t^\circ$ is as defined in \cref{asmp:full}, depending on the payment function.
$b_t^{f_M}$ is never worse than $f_M(v_t)$, since $r_t(b_t^{f_M}) \ge r_t(f_M)$ and $r_t(b_t^{f_M}) \ge 0$, which follow from \cref{asmp:full}.
This guarantees that our reward is always non-negative, which as we discussed in \Cref{subsec:full_info_overview} is the key to bidding according to $\calF_M$ and have $\order*{2^{-M}}$ error compared to using the full $\calF$.

We use the functions in $\calF_0,\calF_1,\ldots,\calF_M$ to create a tree.
For $i=0,\ldots,M$ the $i$-th level of the tree has $|\calF_i|$ nodes, each representing an $f_i \in \calF_i$.
We connect these nodes by defining each node's parent: for every $i = 1, \ldots, M$ and $f_i \in \calF_i$, $P(f_i)$ is the parent of $f_i$ such that $P(f_i) \in \calF_{i-1}$ and $\norm{f_i - P(f_{i-1}) }_\infty \le 2^{-i + 1}$.
We note that such a $P(f_i)$ always exists because of \eqref{eq:full:cover}.
The parent function $P(\cdot)$ subsequently defines the children $\calC(f_i)$ and leaves $\calL(f_i)$ of a node $f_i$, $i = 0,\ldots, M-1$.
Finally, we assume that $|\calF_0| = 1$ (which does not violate \eqref{eq:full:cover}), making the tree have a unique root.

Using this tree, in every round $t$ the algorithm creates a distribution of bids in the following way.
First, each leaf $f_M\in \calF_M$ calculates the bid $b_t^{f_M}$ as defined in \cref{eq:full:good}; this bid defines a trivial bid distribution $q_t^{f_M}(\cdot)$ which suggests the bid $b_t^{f_M}$ with probability $1$.
On the $i$-th level where $i < M$, each non-leaf node $f_i \in \calF_i$ uses a distribution $p_t^{f_i}(\cdot)$ over its children $\calC(f_i)$ and the bid distribution $q_t^{f_{i+1}}(\cdot)$ of each child $f_{i+1} \in \calC(f_i)$ to define its own bid distribution $q_t^{f_i}(\cdot)$.
This node $f_i$ runs an instance of \cref{algo:full:good} to calculate the distribution over its children $p_t^{f_i}(\cdot)$.
This recursive calculation defines the bid distribution of the root, $q_t^{f_0}(\cdot) = q_t(\cdot)$, which is the output of the algorithm.

We show the full algorithm in \cref{algo:full:tree}, which has one slight modification compared to the process described above: each non-leaf node considers one more bid in addition to the bids suggested by its children.
This bid is the maximum bid suggested by any of its leaves.
This additional bid is key to the guarantee of our algorithm since it is the action that is
$\Delta$-good (recall \cref{def:full:good}).
This allows us to use the improved regret bound of \cref{thm:full:good1} which is crucial: a node on the $i$-th level can have as many as $\Theta(\exp(2^i))$ children, which, if $\Delta$ was a constant, would lead to terrible regret bounds.
Instead, we show that the additional action makes $\Delta \approx 2^{-i}$ which makes the regret of every non-leaf node $\tilde O(\sqrt T)$.
%We present the full algorithm in \cref{algo:full:tree}.

\begin{algorithm}[!ht]
% \SetAlgoNoLine
\DontPrintSemicolon
\caption{Tree algorithm}
\label{algo:full:tree}
\KwIn{Number of rounds $T$, Lipschitz parameter $L$, number of levels of tree $M$}

For every $i = 0, 1, \ldots, M$, let $\calF_i$ be as described in \eqref{eq:full:cover}\;

\medskip
\tcp{Description of tree structure}
Let $P(f_i) \in \calF_{i-1}$ such that $\norm{P(f_i) - f_i}_\infty \le 2^{- i + 1}$, $\forall i=1,\ldots,M$, $f_i \in \calF_i$
\tcp*[f]{Parent of $f_i$}

\For{$i = M-1, M-2, \ldots, 0$ and $f_i \in \calF_i$}
{   
    Let $\calC(f_i) = \{f_{i+1} \in \calF_{i+1} : f_i = P(f_{i+1})\}$
    \tcp*[f]{Children of $f_i$}
    
    Let $\calL(f_i) = \bigcup_{f_{i+1} \in \calC(f_i)} \calL(f_{i+1})$
    \tcp*[f]{Leaves of $f_i$}
    
    Let $\calA(f_i)$ be an instance of \cref{algo:full:good}, $\Delta = 2^{-i+3}$ and $K = |\calC(f_i)| + 1$
    \tcp*[f]{Algorithm for $f_i$, with actions its children and the good action}
}

\For{$t \in [T]$}
{   
    Receive $v_t, \chi_t, \psi_t$ and calculate $U_t = \max_{\tau \le t}\{\chi_\tau, \psi_\tau\}$
    \tcp*[f]{$U_t$ is the upper bound of rewards}
    
    For every $f_M \in \calF_M$ calculate $b_t^{f_M}$ as in \cref{eq:full:good}
    \tcp*[f]{Improve bid $f_M(v_t)$ using safe bid}
    
    \hspace{4pt} and let $q_t^{f_M}(\cdot)$ be the bid distribution that bids $b_t^{f_M}$ with probability $1$
    \tcp*[f]{$f_M$'s suggestion}
    
    \medskip
    \tcp{Recursive construction of bid distributions}
    \For{$i = M-1, M-2, \ldots, 0$ and $f_i \in \calF_i$}
    {   
        Calculate $g_t^{f_i} = \max_{f_M \in \calL(f_i)} b_t^{f_M}$
        \tcp*[f]{$f_i$'s good bid}
        
        Get $p_t^{f_i}(\cdot)$ from $\calA(f_i)$ after passing $U_t$
        \tcp*[f]{Distribution over $\calC(f_i) \cup \big\{g_t^{f_i}\big\}$}
        
        Calculate bid distribution $q_t^{f_i}(\cdot)$ over bids $\big\{ b_t^{f_M} \big\}_{f_M \in \calL(f_i)}$:
        \begin{equation*}
            q_t^{f_i}\left( b_t^{f_M} \right)
            \;=\;
            p_t^{f_i}\left( g_t^{f_i} \right) \One{g_t^{f_i} = b_t^{f_M}}
            +
            \sum\nolimits_{f_{i+i} \in \calC(f_i)} p_t^{f_i}(f_{i+1}) q_t^{f_{i+1}}\left( b_t^{f_M} \right)
        \end{equation*}
    }
    
    Sample and use bid $b_t \sim q_t^{f_0}(\cdot) = q_t(\cdot)$
    \tcp*[f]{Sample bid according to the root}
    
    Receive $d_t$
    \tcp*[f]{Get $d_t$, making $r_t(\cdot)$ calculable}
    
    \medskip
    \tcp{Update algorithms of non-leaf nodes}
    \For{$i = 0, 1, \ldots, M$ and $f_i \in \calF_i$}
    {   
        Let $\tilde r_t\big( g_t^{f_i} \big) = r_t\big( g_t^{f_i} \big)$
        \tcp*[f]{Reward of good bid}
        
        Let $\tilde r_t\qty(f_{i+1}) = \sum_b q_t^{f_{i+1}}(b) r_t(b)$, for every $f_{i+1} \in \calC_{i+1}$
        \tcp*[f]{Reward of each child of $f_i$}
        
        Pass above $\tilde r_t(\cdot)$ to $\calA(f_i)$
        \tcp*[f]{Update $f_i$'s algorithm}
    }
}
\end{algorithm}

We now state the guarantee of \cref{algo:full:tree} in \cref{thm:full:tree}.

\begin{theorem}\label{thm:full:tree}
    Let $U = \max_t\{\chi_t, \psi_t\}$ and $\calF$ be the set of $L$-Lipschitz bidding functions for some $L \ge 1$.
    \cref{algo:full:tree} with $M = \lfloor \log_2\sqrt T \rfloor$ generates bid distributions $q_t(\cdot)$ that with probability $1$,
    \begin{equation*}
        \sup_{f \in \calF} \sum_{t=1}^T r_t(f)
        -
        \sum_{t=1}^T \sum_b q_t(b) r_t(b)
        \le
        \order{
            U \sqrt{L T} \log T
        }
        .
    \end{equation*}
\end{theorem}

Our proof is similar to the one in \cite{DBLP:journals/corr/HanZFOW20}, where the most important point is that the effect of the exponential number of actions is mitigated by the small $\Delta$.
One crucial difference that unless we use the modification in \eqref{eq:full:good}, we are in danger of using a bid that would lead to negative reward.
As we explained in \cref{subsec:full_info_overview}, this could lead to significant discretization error because we are bidding using $\calF_M$ instead of $\calF$.

\subsection{Reduction from Regret to Interval Regret}\label{ssec:full:interval}

We now show how to turn the regret bound of \cref{thm:full:tree} to an interval regret bound.
Recall that an interval regret bound is required to apply \cref{thm:pre:primal_dual}.
In this section we prove a general reduction from regret to interval regret, for any problem of online learning with full information feedback.
We first assume that there is an algorithm that, when started on round $\tau_1$, can get some regret $\reg_{\tau_1}(\tau_2)$ by every subsequent round $\tau_2$ (similar to \cref{thm:full:good1}).
Then we create a meta-algorithm that combines multiple instances of the previous algorithms to get interval regret at most $\reg_{\tau_1}(\tau_2) + \order*{\sqrt{T \log T}}$ for every interval $[\tau_1, \tau_2]$.

The meta-algorithm combines the algorithms $\big\{ \calA_{\tau_1} \big\}_{\tau_1 \in [T]}$ as follows.
In round $t$, algorithms $\calA_{\tau_1}$ for which $\tau_1 \le t$ are called \emph{active} and the rest are called \emph{inactive}.
The meta-algorithm receives a distribution $q_t^{\tau_1}(\cdot)$ over actions from each active algorithm and combines them into its own distribution over actions $q_t(\cdot)$, using a distribution over algorithms $p_t(\cdot)$.
The distribution $p_t$ is over all algorithms, active and inactive, and is calculated using a no-regret algorithm $\calA$.
In order to calculate $q_t$ from $p_t$ using only active algorithms, we use $p_t$ conditioned on the active algorithms.

The most interesting part of the meta-algorithm is what rewards are used to update the internal state of $\calA$, the algorithm that outputs $p_t$.
For active algorithms $\calA_{\tau_1}$, $\calA$ uses the reward of the distributions over actions that this algorithm suggests: $\Exs[a \sim q_t^{\tau_1}]{r_t(a)}$.
In contrast, $\calA$ cannot use any information from the inactive algorithms to calculate their rewards.
Instead, $\calA$ uses the expected reward of its own distribution over actions: $\Exs[a \sim q_t]{r_t(a)}$.
This update corresponds to each inactive algorithm outputting $q_t$, the mixture of the distributions of active algorithms.
This approach avoids changing the number of actions that the algorithm $\calA$ considers every round, as done by previous works, e.g. \cite{DBLP:journals/eccc/HazanS07}.
This means that $\calA$ can be an algorithm that utilizes instance-specific properties to get better regret.
For example, since our guarantees need to apply for unknown time-varying reward ranges, we set $\calA$ to be \cref{algo:full:good}.
The full description of our technique is in \cref{algo:full:interval}, which can be found in \cref{ssec:app:full:interval}.

We now state the guarantee of \cref{algo:full:interval}, which offers in every interval regret at most $\tilde O(\sqrt T)$ worse than the one used by the sub-algorithms, $\calA_{\tau_1}$.

\begin{theorem}\label{thm:full:interval}
    Let $A$ be a set of actions and $r_t: A \to [0, U_t]$ be the reward function picked by an adaptive adversary.
    For every $\tau_1 \in [T]$ let $\calA_{\tau_1}$ be an algorithm that, if it starts running from round $\tau_1$, generates distributions $\{ q_t^{\tau_1}(\cdot)\}_{t \ge \tau_1}$ over the actions $A$ such that
    \begin{equation*}
        \forall \tau_2 \ge \tau_1 : \quad
        \sup_{a \in A} \sum_{t \in [\tau_1, \tau_2]} r_t(a)
        -
        \sum_{t \in [\tau_1, \tau_2]} \Exs[a\sim q_t^{\tau_1}]{r_t(a)}
        \le
        \reg_{\tau_1}(\tau_2)
    \end{equation*}
    Then the distributions $q_1(\cdot), \ldots, q_T(\cdot)$ generated by \cref{algo:full:interval} have no interval regret:
    \begin{equation*}
        \forall [\tau_1, \tau_2] \sub [T] : \quad
        \sup_{a \in A} \sum_{t \in [\tau_1, \tau_2]} r_t(a)
        -
        \sum_{t \in [\tau_1, \tau_2]} \Exs[a\sim q_t]{r_t(a)}
        \le
        4 U_{\tau_2} \sqrt{T \log T}
        +
        \reg_{\tau_1}(\tau_2)
        .
    \end{equation*}
\end{theorem}

Combining the above theorem with \cref{thm:full:tree} gives the following result.

\begin{corollary}\label{cor:full:primal_expectation}
    In the same setting as \cref{thm:full:tree}, there is an algorithm that generates distributions over bids $q_1(\cdot), \ldots, q_T(\cdot)$ such that with probability $1$, for every $1 \le \tau_1 < \tau_2 \le T$ it holds
    \begin{equation*}
        \sup_{f \in \calF} \sum_{t=\tau_1}^{\tau_2} r_t(f)
        -
        \sum_{t=\tau_1}^{\tau_2} \sum_b q_t(b) r_t(b)
        \le
        \order{
            U_{\tau_2} \sqrt{L T} \log T
        }
    \end{equation*}
\end{corollary}

With \Cref{cor:full:primal_expectation}, \cref{thm:full:primal} can be proved using concentration inequalities and union bounds. We include these calculations in \cref{sec:app:full} for completeness.

%% file: body/4.tightROI.tex
\section{Exact satisfaction of the ROI constraint}\label{sec:tight_roi}

In this section we show how we can turn every algorithm that has an approximate satisfaction of the ROI constraint into one which has an exact satisfaction.
This reduction (\cref{lem:tight:reduction}), together with \cref{thm:full:main}, will lead to the following theorem.

\begin{theorem}\label{thm:tight:main_tight}
    In the same setting as \cref{thm:full:main}, there exists an algorithm that satisfies the budget and ROI constraints with probability $1$ and for every $\d > 0$, with probability $1 - \d$ has regret
    \begin{equation*}
        \frac{1}{\b^2 \rho}\order{
            \sqrt{L T} \log T
            +
            \sqrt{T \log(1/\d)}
        }
    \end{equation*}
    as long as $\b = \Omega\qty\big( T^{-1/2 + \e}( \sqrt{L} + \sqrt{\log(1/d)}) )$ for some constant $\e > 0$.
\end{theorem}

\cref{thm:tight:main_tight} follows from the following lemma, which combines two algorithms to get the desired reduction from approximate satisfaction of the ROI constraint to an exact one.
One algorithm maximizes the objective and has low ROI violation with high probability.
The other algorithm is much simpler: it maximizes value minus payment, i.e., the ROI constraint.
By running the second algorithm for enough rounds, we can accumulate enough `slack' to mitigate the violation caused by the first algorithm.
We formalize this in the following lemma.

\begin{lemma}\label{lem:tight:reduction}
    Assume that there is an algorithm $\calA_1$ such that for every $\d > 0$, with probability at least $1 - \d$, when running on any set of rounds $\calT_1 \sub [T]$ it generates bids that have
    \begin{itemize}[leftmargin=20pt]
        \item regret at most $\reg_1(|\calT_1|, \d)$.
        \item total ROI violation at most $V(|\calT_1|, \d)$.
    \end{itemize}
    and another algorithm $\calA_2$ such that
    \begin{itemize}[leftmargin=20pt]
        \item its bid $b_t$ in round $t$ satisfies $\One{b_t \ge d_t}(v_t - p(b_t, d_t)) \ge 0$.
        \item for every $\d > 0$, with probability at least $1 - \d$, when running in any set of rounds $\calT_2 \sub [T]$ it generates bids $\{b_t\}_{t\in\calT_2}$ such that $\sum_{t \in \calT_2} \One{b_t \ge d_t}(v_t - p(b_t, d_t)) \ge Q(|\calT_2|, \d)$.
    \end{itemize}
    Assume that $\reg(\cdot, \d), V(\cdot, \d), Q(\cdot, \d)$ are increasing functions for every $\d$.
    Consider the algorithm that bids $0$ on rounds where the remaining budget is less than $1$, uses $\calA_1$ on rounds $t$ such that $\sum_{t' \le t-1} \One{b_t \ge d_t}(v_t - p(b_t, d_t)) \ge 1$,
    and $\calA_2$ on other rounds. The algorithm yields exact satisfaction of the ROI constraint. Moreover, for any $\d > 0$, with probability at least $1 - 2\d$ it has regret at most
    \begin{equation*}
        \reg_1\qty( T , \d ) 
        +
        2 Q^{-1}\qty\big( V(T, \d) + 2)
    \end{equation*}
    where $Q^{-1}(V(T, \d) + 2, \d)$ is the $x$ in the equation $Q(x, \d) = V(T, \d) + 2$.
\end{lemma}

Understanding the conditions of the two algorithms is key to the proof of the lemma.
The first algorithm violates the ROI constraint by at most $V(T)$ (for simplicity we drop the dependence on $\d$) and the second algorithm, in order to make up that violation, needs to be run for about $Q^{-1}(V(T))$ rounds.
This has two effects on the total regret.
First, $\calA_1$ is run for $Q^{-1}(V(T))$ fewer rounds, potentially missing out on any reward on those rounds.
Second, $\calA_2$ can use at most $Q^{-1}(V(T))$ of the budget, making the overall algorithm have to stop at most $Q^{-1}(V(T))$ rounds earlier.
This entails the desired bound.

We briefly explain how we get \cref{thm:tight:main_tight} from \cref{lem:tight:reduction}.
It is obvious that the algorithm of \cref{thm:full:main} satisfies the constraints for $\calA_1$.
For $\calA_2$ we can use the algorithm of \cref{thm:full:primal}: if we set $\chi_t = \psi_t = 1$ in $r_t(\cdot)$ of \eqref{eq:full:reward} we get the desired objective.
This makes $ \reg(T_1) = V(T_1) = \frac{1}{\rho \b} \, \tilde O \qty( T_1 ) $ and $ Q(T_2, \d) = \b T_2 - \tilde O(T_2)$; the last equality follows from the assumption on $\beta$, as shown in Equation~\eqref{eq:pre:beta}.
The theorem's regret guarantee follows by the lower bound on $\b$ in the statement of the theorem. We include the proof of \cref{lem:tight:reduction} and \cref{thm:tight:main_tight} in \cref{sec:app:tight}.

\boldparagraph{Necessity of \texorpdfstring{$\b$}{beta}}
%In this section we present a simple theorem that shows why we require $\b > 0$ in \eqref{eq:pre:beta}.
We present an example to show that the dependency on $\b$ is necessary in our regret bounds. More specifically, we show that as $\b$ decreases, the regret of any algorithm that exactly satisfies the ROI constraint increases polynomially in $1/\b$. We show this for second-price auctions, to simplify the problem of regret minimization.

\begin{theorem}\label{thm:tight:beta}
    There is an example in second-price auctions where any algorithm that always satisfies the ROI constraint has regret compared to the optimum LP value (e.g., see \eqref{eq:pre:opt}) at least $(1 - 2\b)\sqrt{T \frac{1}{2\pi\b(1-\b)}}$ for every constant $\b < 1/2$.
\end{theorem}

%% file: body/5.bandit.tex
\section{Bandit information} \label{sec:bandit}

In this section we study online learning in the bandit information setting.
After bidding on round $t$, the bidder does not observe the highest competing bid $d_t$.
Instead, she only gets to observe a boolean value that indicates whether she wins this round, along with the price she paid if she wins, i.e., $\One{b_t \ge d_t}$ and $\One{b_t \ge d_t} p(b_t, d_t)$.
Note that this setting is completely different between first-price and second-price auctions.
In second-price, if the bidder wins a certain round, she gets to observe $d_t$ and the same is true whenever $p(b, d) = q b + (1-q) d$ for $q > 0$.
In contrast, when the bidder participates in strictly first-price auctions, she can never observe $d_t$.
In \cref{ssec:bandit:lower} we show that this is a crucial distinction that makes first-price auctions much harder: any algorithm has to suffer $\Omega(T^{2/3})$ regret, which is in contrast to the $\tilde O(\sqrt T)$ regret bounds for second-price auctions \cite{DBLP:conf/sigecom/BalseiroG17,DBLP:conf/www/FengPW23}.
To complement the negative result, in \cref{ssec:bandit:upper} we offer an algorithm with $\tilde O(T^{3/4})$ regret.

\subsection{Regret Lower Bound for Bandit Information in First-price Auctions} \label{ssec:bandit:lower}

In this section we show our lower bound on the regret for first-price when only bandit information is available.
An $\Omega(T^{2/3})$ bound is known for quasi-linear maximization and no constraints by \cite{DBLP:conf/nips/BalseiroGMMS19}, which follows from the matching lower bound of \cite{DBLP:conf/focs/KleinbergL03}.
Instead, we show a lower bound for value maximization when there is only a budget constraint.
In particular, our lower bound holds even if the value is fixed across rounds.

\begin{theorem}\label{thm:bandit_lower_bound}
    There exists an instance in first-price auctions with no ROI constraint, $\rho = 1/4$, and $v_t = 1$ such that, any value-maximizing algorithm with only bandit feedback have regret $\Omega(T^{2/3})$.
\end{theorem}

The problem described in \cref{thm:bandit_lower_bound} seems straightforward at first glance since the buyer wants to maximize the number of wins while adhering to the budget constraint.
For example, if the CDF of the distribution that generates $d_t$ is continuous and strictly increasing, one could consider that the bid $b^*$ such that $b^*\cdot \Pr{b^* \ge d_t} = \rho$, i.e., the bid that depletes the budget in expectation, is optimal.
The problem of (approximately) finding such a bid $b^*$ is not hard, since it can be reduced to the problem of noisy binary search \cite{DBLP:conf/soda/KarpK07}.
However, while this approach would work for second-price auctions, it does not work for first-price auctions.
The reason is that our original hypothesis, that a single bid every round is optimal, is erroneous.
There are strictly increasing and continuous CDF for $d_t$ for which, to maximize the number of wins while staying within budget, the bidder needs to alternate between two bids.
In particular, the value of using the optimal two bids and the single optimal bid can be as big as a factor of $2$.

In our proof we study a CDF $F$ for $d_t$ where there are infinite pairs of bids that, if mixed properly, attain the optimal value within budget.
This is the key for our next step, where an adversary perturbs $F$, by moving some mass to some bid $b^*$ from bid $b^*+\e$, making bid $b^*$ more `valuable'.
This entails that any algorithm that does not use $b^*$ enough faces considerable regret.
The proof is completed by arguments similar to the lower bound of \cite{DBLP:conf/focs/KleinbergL03}: the bandit information that the bidder receives requires any algorithm to waste multiple rounds using sub-optimal bids before finding $b^*$.
We include the full proof in \cref{sec:app:bandit:lower}.

\subsection{Regret Upper Bound} \label{ssec:bandit:upper}

In this section we prove an upper bound on the regret when learning with bandit feedback.
Our regret bound is in the order of $\tilde O(T^{3/4})$, which is slightly higher than the lower bound in \cref{thm:bandit_lower_bound}. It remains to be interesting open question to achieve the near-optimal $\tilde O(T^{2/3})$ regret bound.

\begin{theorem}\label{thm:bandit_ub}
    There is a polynomial time algorithm for value or quasi-linear utility maximization when auction satisfies \cref{asmp:full} such that, for every $\d > 0$ with probability at least $1 - \d$ it has regret and ROI violation at most
    %\begin{equation*}
        $\frac{1}{\b \rho} \order{
            \qty(  L^{1/4} T^{3/4} + T^{1/2} L^{1/2} ) \log \frac{L T}{\d}
        }$.
    %\end{equation*}
    %
    In addition, if $\b = \Omega\qty\big( T^{-1/4 + \e}( L^{1/4} + L^{1/2}T^{-1/4})\sqrt{\log(T L/d)} )$ for some constant $\e > 0$, the above can be turned in an algorithm with exact ROI satisfaction regret that is $1/\b$ times worse.
\end{theorem}

\Cref{thm:bandit_ub} is based on the following theorem, \cref{thm:bandit:gen_algo}, which offers an algorithm with $\tilde O(\sqrt{T K})$ high probability interval regret bound.
To use this algorithm, we discretize the interval $[0, 1]$ into $N$ values $\{i/N\}_{i\in[N]}$ and $K$ bids $\{j/K\}_{j\in[K]}$.
Then for each $i \in [N]$ we run an instance of the algorithm of \cref{thm:bandit:gen_algo}, which is used for round $t$ when $v_t \le \tilde i/N < v_t + 1/N$ and outputs a bid $j_t/K$.
As in \cref{sec:full}, we do not directly use the bid suggested by these algorithms, since it might lead to a negative reward in $r_t(\cdot)$; we use the safe bid of \cref{asmp:full} when it is possible to get a negative reward.
This roughly leads to a total regret bound of $\tilde O(\sqrt{T N K})$, along with a discretization error of $O(T(L/N + 1/K))$.
Appropriately picking $N,K$ gives the regret bound.

\cref{thm:bandit:gen_algo} provides the $U_{\tau_2}\tilde O(\sqrt{T K})$ interval regret bound for every interval $[\tau_1, \tau_2]$.
The algorithm used is similar to the EXP-SIX algorithm that \cite{DBLP:conf/nips/Neu15} uses to bound the regret of the best sequence of actions and that \cite{DBLP:journals/corr/CastiglioniCK23} use to bound the interval regret.
However, our algorithm is not the same as EXP-SIX: we modify the algorithm to get a regret bound that scales linearly with $U_{\tau_2}$ instead of $U_{\tau_2}^2$ as shown in \cite{DBLP:journals/corr/CastiglioniCK23}.

\begin{theorem}\label{thm:bandit:gen_algo}
    Suppose there are $K$ actions and the reward of round $t$, $r_t : [K] \to [0, U_t]$, is picked by an adaptive adversary. There exists an algorithm that generates actions $a_1, \ldots, a_T$ such that for every $\d > 0$, with probability at least $1 - \d$, we have that for all $1 \le \tau_1 < \tau_2 \le T$
    \begin{equation*}
        \max_{a \in [K]} \sum_{t = \tau_1}^{\tau_2} r_t(a)
        -
        \sum_{t = \tau_1}^{\tau_2} r_t(a_t)
        \le
        \order{U_{\tau_2} \cdot 
            \qty( \sqrt{T K} + K ) \log \frac{T K}{\d}
        }
    \end{equation*}
\end{theorem}

%% file: body/6.poly.tex
\section{Polynomial time algorithm for full information feedback} \label{sec:poly}

In our final section we present a polynomial time algorithm that has regret guarantees matching the ones in \cref{thm:full:main,thm:tight:main_tight}, when the value and the highest competing bid are independent. 

\begin{theorem} \label{thm:poly}
    Assume that the value $v_t$ and highest competing bid $d_t$ are sampled independently.
    There is a polynomial time algorithm with full information feedback for value or quasi-linear utility maximization when the payment function satisfies \cref{asmp:full} such that, for every $\d > 0$ with probability at least $1 - \d$, the algorithm has regret against the class of $L$-Lipschitz continuous functions and ROI violation at most $\frac{1}{\b\rho} \order*{ \sqrt{T \log (T L / \d)} }$.
    In addition, if $\b = \Omega\qty\big( T^{-1/2 + \e}\log (T L / \d) )$ for some constant $\e > 0$, it can be turned into an algorithm with exact ROI satisfaction and with regret that is $1/\b$ times worse.
\end{theorem}

The algorithm for the above theorem is similar to the one in \cref{thm:bandit_ub}: discretize the values and run a separate online learning algorithm for each of those values.
The important difference is that, instead of running/updating each algorithm only when we observe its corresponding value, we run every algorithm in every round $t$, even if the value $v_t$ is completely different than the one the algorithm represents.
This allows each algorithm to run for many more rounds, making it `learn' faster.
The assumption that $v_t$ and $d_t$ are independent is crucial since the $d_t$ of every round can be used for every algorithm.
The full proof is in \cref{sec:app:poly}.

We note that the above technique cannot directly be used for bandit information, since we do not observe $d_t$ and therefore the observed reward is dependent on the bidding. \cite{DBLP:conf/nips/BalseiroGMMS19} use this technique, along with the assumption that the value distribution is known, to get regret bounds that would imply to $\tilde O(T^{2/3})$ regret for bandit information.
However, their bounds are only in expectation and for the entire horizon, not every interval.
We leave as future work extending these techniques to get regret matching the one in \cref{thm:bandit_lower_bound}.

%% file: appendix/2.prelim.tex
\section{Deferred proof of Section \ref{sec:prelim}} \label{sec:app:prelim}

We know present the assumption that \cite{DBLP:journals/corr/CastiglioniCK23} make for their results.
They assume that there must exist a distribution $F \in \Delta(\calF)$ of functions such that if the player bids according to $F$ then for some $\a \ge 0$ on expectation:
(a) the payment of the player is no more than $\rho - \a$, and
(b) the value earned by the player is at least the payment plus $\a$.
Formally, $\exists F \in \Delta(\calF)$:
\begin{equation}\label{eq:pre:alpha}
\begin{split}
    \E_{f\sim F}\Ex[(v,d)]{
        \qty\Big( v - p\qty\big( f(v), d) ) \One{f(v) \ge d}
    }
    \ge
    \a
    \;\text{ and }\;
    \E_{f\sim F}\Ex[(v,d)]{
        p\qty\big( f(v), d ) \One{f(v) \ge d}
    }
    \le
    \rho - \a
\end{split}
\end{equation}

The above conditions imply that substituting $F$ in problem \eqref{eq:pre:opt}, both constraints are satisfied with a slack of $\a$.
In the absence of a ROI constraint, we would have $\a = \rho$.
Given \eqref{eq:pre:alpha} and that $\a > 0$, the regret bounds of \cite{DBLP:journals/corr/CastiglioniCK23} depend on $\a$ and become worse as $\a$ becomes smaller.

We now prove how \eqref{eq:pre:beta} implies \eqref{eq:pre:alpha} with $\a = \b\rho$.
We require that the function $f(v) = 0$ is included in $\calF$ and bidding $0$ guarantees $0$ payment.
Let $f_\b$ be the bidding function that makes the guarantee in \eqref{eq:pre:beta}.
We define $F$ to be the following distribution of functions: $f_\b$ with probability $\rho$ and the $0$ bid with probability $1 - \rho$.
In this case we have
\begin{equation*}
    \E_{f\sim F}\Ex[(v,d)]{
        \big( v - p\qty\big( f(v), d) ) \big) \One{f(v) \ge d}
    }
    \ge
    \rho \Ex[(v,d)]{
        \big( v - p\qty\big( f_\b(v), d) ) \big) \One{f_\b(v) \ge d}
    }
    \ge \rho \b
\end{equation*}
where the first inequality holds by $\One{0 \ge d} p(0, d) = 0$ and the second by \eqref{eq:pre:beta}. 
This proves the first inequality of \eqref{eq:pre:alpha}.
For the second constraint of \eqref{eq:pre:alpha} we have
\begin{alignat*}{3}
    \Line{
        \E_{f\sim F}\Ex[(v,d)]{
            p\qty\big( f(v), d) ) \One{f(v) \ge d}
        }
    }{=}{
        \rho \Ex[(v,d)]{
            p\qty\big( f_\b(v), d) ) \One{f_\b(v) \ge d}
        }
    }{}
    \\
    \Line{}{\le}{
        \rho ( 1 - \b )
    }{\text{by \eqref{eq:pre:beta} and } v \le 1}
\end{alignat*}
The above two inequalities prove \eqref{eq:pre:alpha} for $\alpha = \rho \b$, as claimed.

%% file: appendix/3.full.tex
\section{Deferred Proofs and Text of Section \ref{sec:full}} \label{sec:app:full}

In this section we present the deferred proofs and text of \cref{sec:full}.

\subsection{Deferred proofs from Section \ref{ssec:full:good}} \label{ssec:app:full:good}

We first prove \cref{thm:full:good1} the theorem from \cref{ssec:full:good}.

\begin{proof}[Proof of \cref{thm:full:good1}]
    We first shift the rewards: since Hedge remains unchanged if a (possibly dependent on time) constant is added to the rewards, we set for all $t,a$
    \begin{equation*}
        r_t(a)
        \gets
        r_t(a) + \Delta U_t
        - \max_{a'\in [K]} r_t(a')
    \end{equation*}
    This means that now the rewards are $r_t:[K] \to [-(1-\Delta)U_t, \Delta U_t]$ and specifically for the good action, $r_t(g) \ge 0$.

    Now let $W_t = \sum_a \exp(\eta_t R_{t-1}(a))$ and $\theta = \sqrt{\frac{\log L}{T \Delta}}$.
    Notice that the probability to play action $a$ in round $t$ is $p_t(a) = \frac{\exp( \eta_t R_{t-1}(a) )}{W_t}$.
    We have that
    \begin{alignat}{3} \label{eq:full:1}
        \Line{
            \frac{1}{W_t} \sum_a \exp(\eta_t R_t(a))
        }{=}{
            \sum_a \frac{\exp(\eta_t R_{t-1}(a))}{W_t} \exp(\eta_t r_t(a))
        }{}
        \nonumber\\
        \Line{}{=}{
            \sum_a p_t(a) \exp(\eta_t r_t(a))
        }{}
        \nonumber\\
        \Line{}{\le}{
            \exp\left(
                \eta_t(1 - \eta_t U_t) \sum_a p_t(a) r_t(a)
                +
                \eta_t^2 U_t \Delta U_t
            \right)
        }{}
        \nonumber\\
        \Line{}{=}{
            \exp\left(
                \eta_t(1 - \theta) \sum_a p_t(a) r_t(a)
                +
                \theta^2 \Delta
            \right)
        }{\eta_t = \frac{\theta}{U_t}}
    \end{alignat}
    where in order to prove the last inequality we first prove the following proposition.
    
    \begin{proposition}
        Let $X$ be a random variable such that $\Ex{X} = x$, $c_1 \le X \le c_2$, with $c_2 \ge 0$. Then, for any $0 < \eta \le \frac{1}{\max\{|c_1|,|c_2|\}}$,
        \begin{equation*}
            \Ex{\exp(\eta X)}
            \le
            \exp\big(
                \eta x (1 - \eta (c_2 - c_1) ) + \eta^2 c_2 (c_2 - c_1)
            \big)
            .
        \end{equation*}
    \end{proposition}
    
    \begin{proof}
        Let $\s^2 = \Ex{(X-x)^2}$ and $c = \max\{|c_1|,|c_2|\}$.
        We have that
        \begin{alignat*}{3}
            \Line{
                \Ex{\exp(\eta X)}
            }{=}{
                \Ex{\exp(\eta (X - x))} \exp(\eta x)
            }{}
            \\
            \Line{}{\le}{
                \exp\left(
                    \frac{\s^2}{c^2} \left( e^{\eta c} - 1 -\eta c \right)
                \right)
                \exp(\eta x)
            }{}
            \\
            \Line{}{\le}{
                \exp\left(
                    \frac{\s^2}{c^2} \eta^2 c^2
                \right)
                \exp(\eta x)
            }{\eta c \le 1 \implies e^{\eta c} \le 1 + \eta c + (\eta c)^2}
            \\
            \Line{}{=}{
                \exp\left(
                    \s^2 \eta^2
                \right)
                \exp(\eta x)
            }{}
        \end{alignat*}
        where the first inequality follows from Bernstein's inequality\footnote{See Lemma 7.26 in \url{https://www.stat.cmu.edu/~larry/=sml/Concentration.pdf}.}.
        We now bound
        \begin{equation*}
            \s^2
            =
            \Ex{(X-x)^2}
            \le
            \Ex{(c_2 - X)^2}
            \le
            (c_2 - c_1) \Ex{c_2 - X}
            =
            (c_2 - c_1) (c_2 - x)
        \end{equation*}
        where the first inequality follows from the fact that $\Ex{(X - y)^2}$ is minimized when $y = x = \Ex{X}$, i.e., when it is equal to the variance.
        The second inequality follows from $c_2 - X \ge 0$ and $X \ge c_1$.
        Rearranging proves the proposition.
    \end{proof}
    
    Now the inequality in \eqref{eq:full:1} follows from the proposition by setting $X = r_t(a)$ with probability $p_t(a)$, $c_1 = -(1-\Delta) U_t$, and $c_2 = \Delta U_t$ and noticing that $\eta_t = \frac{\theta}{U_t} \le \frac{1}{U_t} = \frac{1}{c_2 - c_1} \le \frac{1}{\max\{|c_1|,|c_2|\}}$.
    
    Taking the logarithm of \eqref{eq:full:1} we get
    \begin{equation*}
        \frac{1}{\eta_t} \log \left(
            \frac{\sum_a \exp(\eta_t R_t(a))}{\sum_a \exp(\eta_t R_{t-1}(a))}
        \right)
        \le
        (1 - \theta) \sum_a p_t(a) r_t(a)
        +
        U \theta \Delta
    \end{equation*}
    or equivalently
    \begin{equation} \label{eq:full:99}
        \frac{1}{\eta_t}
        \log \frac{\sum_a \exp(\eta_t R_t(a))}{K}
        -
        \frac{1}{\eta_t}
        \log \frac{\sum_a \exp(\eta_t R_{t-1}(a))}{K}
        \le
        (1 - \theta) \sum_a p_t(a) r_t(a)
        +
        U \theta \Delta
    \end{equation}
    
    In the above equation we use the fact that
    \begin{equation*}
        \frac{1}{\eta_t}
        \log \frac{\sum_a \exp(\eta_t R_t(a))}{K}
        \ge
        \frac{1}{\eta_{t+1}}
        \log \frac{\sum_a \exp(\eta_{t+1} R_t(a))}{K}
    \end{equation*}
    which follows from the fact that $\eta_{t+1} \ge \eta_t$ and the fact that the function
    \begin{equation*}
        \left(
            \frac{1}{K} \sum_{i=1}^K x_i^\eta
        \right)^{1/\eta}
    \end{equation*}
    is increasing in $\eta$ for $x_i > 0$. This makes the previous inequality
    \begin{equation*}
        \frac{1}{\eta_{t+1}}
        \log \frac{\sum_a \exp(\eta_{t+1} R_t(a))}{K}
        -
        \frac{1}{\eta_t}
        \log \frac{\sum_a \exp(\eta_t R_{t-1}(a))}{K}
        \le
        (1 - \theta) \sum_a p_t(a) r_t(a)
        +
        U \theta \Delta
    \end{equation*}
    
    Fix $\tau \in [T]$.
    We add the above for all $t \in [\tau-1]$ along with \eqref{eq:full:99} for $t=\tau$ and simplify the telescopic sum to get
    \begin{equation*}
        \frac{1}{\eta_\tau}
        \log \frac{\sum_a \exp(\eta_\tau R_\tau(a))}{K}
        -
        \frac{1}{\eta_1}
        \log \frac{\sum_a \exp(0)}{K}
        \le
        (1 - \theta) \sum_{t \in [\tau]}\sum_a p_t(a) r_t(a)
        +
        U \theta T \Delta
    \end{equation*}
    Using the fact that $\sum_a \exp(\eta_\tau R_\tau(a)) \ge \exp(\eta_\tau R_\tau^*)$ (where $R_\tau^* = \max_a R_\tau(a)$) and substituting $\eta_\tau = \frac{\theta}{U_\tau}$ and $\theta = \sqrt{\frac{\log K}{T \Delta}}$ we get
    \begin{equation}\label{eq:full:3}
        R_\tau^*
        -
        U_\tau \sqrt{T \Delta \log K}
        \le
        (1 - \theta) \sum_{t \in [\tau]}\sum_a p_t(a) r_t(a)
        +
        U_\tau \sqrt{T \Delta \log K}
    \end{equation}
    
    Using the fact that $R_\tau^* \ge 0$ (since the reward of the good arm is always non-negative) we can use \eqref{eq:full:3} to prove
    \begin{equation}\label{eq:full:4}
        \sum_{t \in [\tau]}\sum_a p_t(a) r_t(a)
        \ge
        -\frac{2}{1-\theta} U_\tau \sqrt{T \Delta \log K}
        \ge
        - 4 U_\tau \sqrt{T \Delta \log K}
    \end{equation}
    where we use the fact that $\theta = \sqrt{\frac{\log K}{T \Delta}} \le 1/2$ since $\Delta \ge \frac{4\log K}{T}$.
    We rearrange the terms in \eqref{eq:full:3} to get
    \begin{alignat*}{3}
        \Line{
            R_\tau^*
            -
            \sum_{t \in [\tau]}\sum_a p_t(a) r_t(a)
        }{\le}{
            2 U_\tau \sqrt{T \Delta \log K}
            -
            \theta \sum_{t \in [\tau]}\sum_a p_t(a) r_t(a)
        }{}
        \\
        \Line{}{\le}{
            2 U_\tau \sqrt{T \Delta \log K}
            +
            4 \theta U_\tau \sqrt{T \Delta \log K}
        }{\text{using \eqref{eq:full:4}}}
        \\
        \Line{}{\le}{
            4 U_\tau \sqrt{T \Delta \log K}
        }{\theta \le \frac{1}{2}}
    \end{alignat*}
    which proves the theorem.
\end{proof}

We now state and prove the high probability version of \cref{thm:full:good1}.

\begin{theorem}\label{thm:full:good2}
    In the same setting as \cref{thm:full:good1}, \cref{algo:full:good} guarantees the following high probability bound:
    for every $\d > 0$ probability at least $1 - \d$, it holds that
    \begin{equation*}
        \forall \tau \in [T] : \quad
        \max_{a\in[K]} \sum_{t \in [\tau]} r_t(a)
        -
        \sum_{t \in [\tau]} r_t(a_t)
        \le
        4 U_\tau \qty(
            \sqrt{T \Delta \log K}
            +
            \max\left\{
                \sqrt{T \Delta \log(T/\d)}
                ,
                \log(T/\d)
            \right\}
        )
    \end{equation*}
\end{theorem}

The high probability bound does not follow from a simple application of the Azuma-Heoffding inequality, since the range of $\sum_t r_t(a_t)$ can be $\Omega(U T)$ making the resulting error $\order*{U\sqrt{T \log(1/\d)}}$ and not $\order*{\sqrt{T \Delta \log(1/\d)}}$ like in the above.
Instead, we use Freedman's inequality, which offers a bound based on $\sum_t \text{Var}[r_t(a_t)]$ which we prove is $\order*{U^2 T \Delta}$.
This allows us to get the improved dependence on $\Delta$.

\begin{proof}[Proof of \cref{thm:full:good2}]
    For every $t$, let $X_t = \sum_a p_t(a) r_t(a)$ and $Y_t = \sum_{\tau = 1}^t (X_\tau - r_\tau(a_\tau))$.
    The theorem follows by showing that for every $\d > 0$
    \begin{equation*}
        \Pr{
            \forall\tau \in [T] :
            Y_\tau \le
            4 U_\tau \max\left\{
            \sqrt{T \Delta \log(T/\d)}
            ,
            \log(T/\d)
        \right\}
        }
        \ge 1 - \d
        .
    \end{equation*}
    
    We are going to use Freedman's inequality \cite[Theorem 1.6]{freedman1975tail} on the sequence $Y_0, Y_1, \ldots$ which we first prove is a martingale with respect to the the history of the rounds (we denote with $\Es[t]{\cdot}$ the expectation conditioned on the history of the rounds up to $t$, i.e., the actions that the player and the adversary has take up to $t$): for every $t \ge 1$
    \begin{equation*}
        \Exs[t-1]{Y_t - Y_{t-1}}
        =
        \Exs[t-1]{X_t - r_t(a_t)}
        =
        0
    \end{equation*}
    where the last inequality holds because $a_t = a$ with probability $p_t(a)$.
    This proves the martingale property.
    We now notice that $|Y_t - Y_{t-1}| \le U_t$ and that
    \begin{alignat*}{3}
        \Line{
            \Exs[t-1]{(Y_t - Y_{t-1})^2}
        }{=}{
            \Exs[t-1]{(X_t - r_t(a_t))^2}
        }{}
        \\
        \Line{}{\le}{
            \Exs[t-1]{( \Delta U_t - r_t(a_t) )^2}
        }{\Exs[t-1]{r_t(a_t)} = X_t}
        \\
        \Line{}{\le}{
            U_t \Exs[t-1]{\Delta U_t - r_t(a_t) }
            =
            \Delta U_t^2 - U_t X_t
        }{-(1-\Delta)U_t \le  r_t(a) \le \Delta U_t}
    \end{alignat*}
    where we notice that in the first inequality we use the fact that $\Exs[t-1]{(X_t - r_t(a_t))^2}$ is the conditional variance of $r_t(a_t)$, which means that $\Exs[t-1]{(c - r_t(a_t))^2}$ is minimized when $c = \Exs[t-1]{r_t(a_t)} = X_t$. We now have that
    \begin{alignat*}{3}
        \Line{
            \sum_{t \in [\tau]} \Exs[t-1]{(Y_t - Y_{t-1})^2}
        }{\le}{
            \sum_{t \in [\tau]} ( \Delta U_t^2 - U_t X_t )
            =
            U_\tau^2 T  \Delta - U_\tau \sum_{t \in [\tau]} X_t
        }{}
        \\
        \Line{}{\le}{
            U_\tau^2 T  \Delta - U_\tau \left(
                \max_{a\in[K]} \sum_{t = 1}^T r_t(a)
                - 4 U_\tau \sqrt{T \Delta \log K}
            \right)
        }{\text{\cref{thm:full:good1}}}
        \\
        \Line{}{\le}{
            U_\tau^2 T \Delta
            + 4 U_\tau^2 \sqrt{T \Delta \log K}
        }{r_t(g) \ge 0}
        \\
        \Line{}{\le}{
            3 U_\tau^2 T \Delta
        }{\log K \le \frac{T \Delta}{4}}
    \end{alignat*}
    
    Now using Freedman's inequality gives us that for all $\e > 0$
    \begin{equation*}
        \Pr{Y_\tau < \e}
        \ge
        1 - \exp\left(
            - \frac{\e^2 / 2}{3 U_\tau^2 T \Delta + U_\tau \e /3}
        \right)
    \end{equation*}
    
    Let $\d > 0$ such that
    \begin{equation}\label{eq:full:5}
        \e = U_\tau \max\left\{
            \sqrt{12} \sqrt{T \Delta \log(1/\d)}
            ,
            \frac{4}{3} \log(1/\d)
        \right\}
        .
    \end{equation}
    This and a union bound over all $\tau$ proves the lemma as long as we prove that
    \begin{equation*}
        \frac{\e^2 / 2}{3 U_\tau^2 T \Delta + U_\tau \e /3}
        \ge
        \log(1/\d)
    \end{equation*}
    or equivalently
    \begin{equation*}
        \e^2
        \ge
        6 U_\tau^2 T \Delta \log(1/\d) + \frac{2}{3} U_\tau \e \log(1/\d)
    \end{equation*}
    
    The above inequality is true because, by the definition of $\d$ in \eqref{eq:full:5},
    \begin{equation*}
        \e^2
        \ge 
        12 U_\tau^2 T \Delta \log(1/\d)
        \quad
        \text{ and }
        \quad
        \e
        \ge 
        \frac{4}{3} U_\tau \log(1/\d)
    \end{equation*}
    Multiplying the second inequality with $\e$ and adding them gives us the desired bound on $\e^2$.
\end{proof}

\subsection{Deferred proof of Section \ref{ssec:full:tree}} \label{ssec:app:full:tree}

We now prove \cref{thm:full:tree}.

\begin{proof}[Proof of \cref{thm:full:tree}]
    We first make the observation that \cref{algo:full:tree} is well defined: in order to calculate $b_t^{f_M}$ for every $f_M \in \calF_M$ we only need knowledge of $v_t,\chi_t,\psi_t$ and not $d_t$, as explained in \cref{asmp:full}.
    We also note that the rewards $\tilde r_t(\cdot)$ that are fed into each $\calA(f_i)$ are in the range $[0, U_t]$ and $U_1 \le U_2 \le \ldots$, as needed for the guarantee of \cref{thm:full:good1}.
    The lower bound for the rewards comes from the fact that every bid used in round $t$ is $b_t^{f_M}$ for some $f_M \in \calF_M$, which because of \eqref{eq:full:good} and \cref{asmp:full} guarantees non-negative reward.
    The upper bound follows from the definition of $r_t(\cdot)$ and the fact that values and bids are in $[0, 1]$.
    
    We now show that for every $i < M$ and $f_i \in \calF_i$, the good bid, $g_t^{f_i}$, is $2^{-i+3}$-good with respect to the bids used by $f_i$, i.e. $\big\{ b_t^{f_M} \big\}_{f_M \in \calL(f_i)}$.
    We prove that
    \begin{equation}\label{eq:full:15}
        r_t\left( g_t^{f_i} \right)
        \ge
        r_t\left( b_t^{f_M} \right)
        -
        2^{-i+3} U_t,
        \quad
        \forall f_M \in \calL(f_i)
    \end{equation}
    which implies $\tilde r_t\big( g_t^{f_i} \big) \ge \tilde r_t\big( f_{i+1} \big) - 2^{-i+3} U_t$ for all $f_{i+1} \in \calC(f_i)$, since $\tilde r_t\big( g_t^{f_i} \big) = r_t\big( g_t^{f_i} \big)$ and $\tilde r_t\big( f_{i+1} \big)$ is a convex combination of $\{ r_t\big( b_t^{f_M} \big) \}_{f_M \in \calL(f_i)}$.
    
    Fix $f_M \in \calL(f_i)$. We distinguish two cases to prove \eqref{eq:full:15} for this $f_M$:
    \begin{itemize}
        \item If $b_t^{f_M}$ loses the auction in round $t$ ($b_t^{f_M} < d_t$), then $r_t(b_t^{f_M}) = 0$ and \eqref{eq:full:15} follows from $r_t( g_t^{f_i} ) \ge 0$.
        
        \item If $b_t^{f_M}$ wins the auction in round $t$, then $g_t^{f_i} \ge b_t^{f_M}$ and therefore $g_t^{f_i}$ also wins the auction in round $t$.
        This means that $\One{b_t^{f_M} \ge d_t} = \One{g_t^{f_i} \ge d_t} = 1$ and so, in order to prove \eqref{eq:full:15} we have to prove that the payment of $g_t^{f_i}$ is not more than the payment of $b_t^{f_M}$ plus $2^{-i+3}$.
        The last statement follows from the Lipschitzness of $p(\cdot, d_t)$ (\cref{asmp:full}) and the fact that $| g_t^{f_i} - b_t^{f_M} | \le 2^{-i+3}$ which follows by the following:
        For every $f_M, f_M' \in \calL(f_i)$ it holds that $\norm{f_M - f_M'}_\infty \le 2^{-i+3}$ since
        \begin{equation*}
            \norm{ f_i - f_M }_\infty
            \le
            \sum_{j=0}^{M-i-1} \norm{ P^{j+1}(f_M) - P^j(f_M)  }_\infty
            \le
            \sum_{j=0}^{M-i-1} 2^{-M+j+2}
            \le
            2^{-i+2}
        \end{equation*}
        where $P^j(\cdot)$ is the application of the parent function $P$ $j$ times, the first inequality uses the triangle inequality, and the second uses the definition of the the parent function $P$.
        The fact that $\norm{f_M - f_M'}_\infty \le 2^{-i+3}$ follows from the above using the triangle inequality and the fact that $| g_t^{f_i} - b_t^{f_M} | \le 2^{-i+3}$ follows from the fact that $g_t^{f_i} = b_t^{f_M'}$ for some $f_M' \in \calL(f_i)$.
    \end{itemize}
    
    Now we summarize the setting of each algorithm $\calA(f_i)$, for $f_i \in \calF_i, i<M$:
    \begin{itemize}
        % \item A distribution $p_t^{f_t}(\cdot)$ is maintained over $\calC(f_i) \cup \{ g_t^{f_i} \}$.
        % \item The distribution $p_t^{f_t}(\cdot)$ is calculated using the algorithm of \cref{thm:full:good1} and the rewards $\tilde r_t(\cdot)$.
        \item The reward range is $[0, U_t]$ in round $t$, where $U_t = \max\{\chi_t, \psi_t\}$.
        \item In every round there is an action that is $\Delta_i$-good, where $\Delta_i := 2^{-i+3}$.
        \item There are at most $K_i$ actions, where $K_i := \exp(\Clip L 2^{i+1}) + 1 \le \exp(\Clip L 2^{i+2})$ where last inequality holds because $\Clip L \ge 1$.
    \end{itemize}
    
    Let $\tilde R_T(f_i) = \sum_{t=1}^T \tilde r_t(f_i)$ denote the total reward of algorithm $\calA(f_i)$ and similarly define $\tilde R_T(g^{f_i})$ the total reward of the good bids of $\calA(f_i)$.
    Because of the guarantee of each algorithm (\cref{thm:full:good1}) we have that with probability $1$:
    \begin{equation}\label{eq:full:11}
        \max_{f_{i+1} \in \calC(f_i) \cup \{ g^{f_i} \}} \tilde R_T(f_{i+1})
        -
        \sum_{f_{i+1} \in \calC(f_i) \cup \{ g^{f_i} \}} p_t^{f_i}(f_{i+1}) \tilde R_T(f_{i+1})
        \le
        4 U \sqrt{T \Delta_i \log K_i}
        % +
        % 4 U \sqrt{T \Delta_i \log(1/\d_i)}
        % +
        % 4 U \log(1/\d_i)
        \le
        23 U \sqrt{\Clip L T}
    \end{equation}
    
    % We now notice that, using the union bound on all the algorithms, with probability at least $1 - \sum_{i=0}^{M-1} \d_i \exp(\Clip L 2^i)$ the above holds for all $i,f_i$. For any $\d >0$, picking $\d_i = \frac{\d}{M \exp(\Clip L 2^i)}$ we get that with probability $1 - \d$ for all algorithms it holds
    % \begin{alignat}{3}\label{eq:full:11}
    %     \Line{
    %         \max_{f_{i+1}\in\calC(f_i)} R(f_{i+1})
    %         -
    %         R(f_i)
    %     }{\le}{
    %         16 U  \sqrt{T \Clip L}
    %         +
    %         8 U \sqrt{T (2^{-i} \log\frac{M}{\d} + \Clip L)}
    %         +
    %         4 U (\log\frac{M}{\d} + \Clip L 2^i)
    %     }{}
    %     \nonumber\\
    %     \Line{}{\le}{
    %         24 U \sqrt{T \Clip L}
    %         +
    %         8 U \sqrt{T 2^{-i} \log\frac{M}{\d}}
    %         +
    %         4 U \log\frac{M}{\d}
    %         +
    %         4 U \Clip L 2^i
    %     }{}
    % \end{alignat}
    
    We now bound the error because we bid according to the bidding functions $\calF_M$ and not $\calF$.
    For any $f \in \calF$ let $f_M \in \calF_M$ be such that $f_M \ge f$ and $\Vert f - f_M \Vert_\infty \le 2^{-M}$.
    For every round $t$ we have
    \begin{alignat*}{3}
        \Line{
            r_t(f)
        }{=}{
            \One{f(v_t) \ge d_t}\Big( \chi_t v_t - \psi_t p\big( f(v_t), d_t \big) \Big)
        }{}
        \\
        \Line{}{\le}{
            \One{f_M(v_t) \ge d_t}\Big( \chi_t v_t - \psi_t p\big( f(v_t), d_t \big) \Big)^+
        }{f_M \ge f}
        \\
        \Line{}{\le}{
            \One{f_M(v_t) \ge d_t}\Big( \chi_t v_t - \psi_t p\big( f_M(v_t), d_t \big) \Big)^+
            +
            \psi_t f_M(v_t) - \psi_t f(v_t) 
        }{p(\cdot,d_t): 1\text{-Lipschitz}}
        \\
        \Line{}{\le}{
            \One{f_M(v_t) \ge d_t}\Big( \chi_t v_t - \psi_t p\big( f_M(v_t), d_t \big) \Big)^+
            +
            U_t 2^{-M}
        }{\Vert f - f_M \Vert_\infty \le 2^{-M}}
        \\
        \Line{}{=}{
            \big( r_t(f_M) \big)^+
            +
            U_t 2^{-M}
            \le
            \tilde r_t(f_M)
            +
            U_t 2^{-M}
        }{\tilde r_t(f_M) \ge ( r_t(f_M) )^+}
    \end{alignat*}
    where recall that $\tilde r_t(f_M) = r_t(b_t^{f_M})$.
    
    The above implies
    \begin{equation}\label{eq:full:12}
        % \opt
        % =
        \sup_{f\in\calF} \sum_{t = 1}^T r_t(f)
        \le
        \max_{f_M\in\calF_M} \tilde R_t(f_M)
        + U T 2^{-M}
    \end{equation}
    
    Let $f_M^*$ be a maximizer of the r.h.s. in the inequality above and for every $i < M$, let $f_i^*$ be the ancestor of $f_M^*$ in the $i$-th level.
    Using this notation we prove
    \begin{alignat*}{3}
        \Line{
            \sup_{f\in\calF} \sum_{t = 1}^T r_t(f)
            -
            \tilde R_T(f_0)
        }{\le}{
            U T 2^{-M}
            +
            \tilde R_T(f_M^*)
            -
            \tilde R_t(f_0)
        }{\text{by \eqref{eq:full:12}}}
        \\
        \Line{}{=}{
            U T 2^{-M}
            +
            \sum_{j=0}^{M-1}\left(
                \tilde R_T(f_{j+1}^*)
                -
                \tilde R_T(f_j^*)
            \right)
        }{}
        \\
        \Line{}{\le}{
            U T 2^{-M}
            +
            \sum_{j=0}^{M-1} 23 U \sqrt{\Clip L T}
        }{\text{by \eqref{eq:full:11}}}
        \\
        \Line{}{=}{
            U T 2^{-M}
            +
            23 U M \sqrt{\Clip L T}
        }{}
    \end{alignat*}
    Picking $M = \lfloor \log_2\sqrt T \rfloor$ the above becomes
    \begin{equation*}
        \sup_{f\in\calF} \sum_{t = 1}^T r_t(f)
        -
        \tilde R_T(f_0)
        \le
        2 U \sqrt{T}
        +
        \frac{23}{2 \log 2} U \sqrt{\Clip L T} \log T
    \end{equation*}
    The above is the claimed regret bound since $\tilde r_t(f_0) = \sum_b q_t^{f_0} r_t(b)$.
\end{proof}

\subsection{Deferred proof and algorithm from Section \ref{ssec:full:interval}} \label{ssec:app:full:interval}

We first present the reduction from standard no-regret to no interval regret.

\begin{algorithm}[t]
% \SetAlgoNoLine
\DontPrintSemicolon
\caption{Reduction from regret to interval regret}
\label{algo:full:interval}
\KwIn{Number of rounds $T$, action space $A$, algorithms $\big\{ \calA_{\tau_1} \big\}_{\tau_1 \in [T]}$ over action space $A$}

Initialize an instance $\calA$ of \cref{algo:full:good} with $\Delta = 1$ and $K = T$\;

\For{$t \in [T]$}
{
    Receive reward range $[0, U_t]$\;
    
    \For{$\tau_1 \le t$}{
        Pass $U_t$ to $\calA_{\tau_1}$ and
        receive $q_t^{\tau_1}(\cdot)$
        \tcp*[f]{Distribution over actions of $\calA_{\tau_1}$}
    }
    
    Pass $U_t$ to $\calA$ and receive $p_t(\cdot)$
    \tcp*[f]{Distribution over algorithms of $\calA$}
    
    Calculate for every action $a$: $q_t(a) = \frac{\sum_{\tau_1\le t} p_t(\tau_1) q_t^{\tau_1}(a)}{\sum_{\tau_1 \le t} p_t(\tau_1)}$
    \tcp*[f]{Distribution of actions by sampling an algorithm $\calA_{\tau_1}, t \le \tau_1$ and then an action from $q_t^{\tau_1}(\cdot)$}
    
    Sample and output $a_t \sim q_t(\cdot)$\;
    
    Receive function $r_t : A \to [0, U_t]$\;
    
    Pass $r_t(\cdot)$ to $\calA_{\tau_1}$ for $\tau_1 \le t$
    \tcp*[f]{$\calA_{\tau_1}$ internal update}
    
    Calculate $\tilde r_t(\tau_1) = \Exs[a \sim q_t^{\tau_1}]{r_t(a)}$ for $\tau_1 \le t$
    \tcp*[f]{Expected reward of $q_t^{\tau_1}(\cdot)$}
    
    Calculate $\tilde r_t^\emptyset = \Exs[a \sim q_t]{r_t(a)}$
    \tcp*[f]{Expected reward of $q_t(\cdot)$}
    
    Update $\calA$ with reward $\tilde r_t(\tau_1)$ for $\tau_1 \le t$
    \tcp*[f]{$\calA$ update for active algorithms}
    
    \hspace{8pt} and with reward $\tilde r_t^\emptyset$ for $\tau_1 \le t$
    \tcp*[f]{$\calA$ update for inactive algorithms}
}
\end{algorithm}

We now prove \cref{thm:full:interval}.

\begin{proof}[Proof of \cref{thm:full:interval}]
    We first extend the definition of $\tilde r_t(\cdot)$ for inactive algorithms. This makes
    \begin{equation*}
        \tilde r_t(\tau_1)
        =
        \begin{cases}
            \Exs[a \sim q_t^{\tau_1}]{ r_t(a) }, & \text{ if } \tau_1 \le t
            \\
            \tilde r_t^\emptyset, & \text{ if } \tau_1 > t
        \end{cases}
        .
    \end{equation*}
    
    Now we re-write $\tilde r_t^\emptyset$, the expected reward if the action is sampled according to $q_t(\cdot)$:
    \begin{equation}\label{eq:full:52}
        \tilde r_t^\emptyset
        =
        \Exs[a \sim q_t]{r_t(a)}
        =
        \frac{1}{
            \sum_{\tau_1 \le t} p_t(\tau_1)
        }
        \sum_{\tau_1 \le t} \qty\big(
            p_t(\tau_1) \tilde r_t(\tau_1)
        )
        .
    \end{equation}
    
    We notice that $\tilde r_t(\cdot)$ has the same reward range as $r_t(\cdot)$.
    \cref{thm:full:good1} for algorithm $\calA$ gives us a regret guarantee by every round $\tau_2$:
    \begin{equation*}
        \forall \tau_2 \in [T] : \quad
        \max_{\tau_1 \in [T]} \sum_{t\in[\tau_2]} \tilde r_t(\tau_1)
        -
        \sum_{t\in[\tau_2]} \sum_{\tau_1 \in [T]} p_t(\tau_1) \tilde r_t(\tau_1)
        \le
        4 U_{\tau_2} \sqrt{T \log T}
    \end{equation*}
    which implies
    \begin{equation}\label{eq:full:31}
        \forall [\tau_1, \tau_2] \sub [T] : \quad
        \sum_{t\in[\tau_2]} \tilde r_t(\tau_1)
        -
        \sum_{t\in[\tau_2]} \sum_{\tau_1 \in [T]} p_t(\tau_1) \tilde r_t(\tau_1)
        \le
        4 U_{\tau_2} \sqrt{T \log T}
    \end{equation}
    
    We are going to show that \eqref{eq:full:31} implies our theorem. First we prove that for every round $t$,
    \begin{alignat*}{3}
        \Line{
            \sum_{\tau_1 \in [T]} p_t(\tau_1) \tilde r_t(\tau_1)
        }{=}{
            \sum_{\tau_1 \le t} p_t(\tau_1) \tilde r_t(\tau_1)
            +
            \sum_{\tau_1 > t} p_t(\tau_1) \tilde r_t^\emptyset
        }{}
        \\
        \Line{}{=}{
            \tilde r_t^\emptyset \sum_{\tau_1 \le t} p_t(\tau_1)
            +
            \tilde r_t^\emptyset \sum_{\tau_1 > t} p_t(\tau_1)
        }{\text{by \eqref{eq:full:52}}}
        \\
        \Line{}{=}{
            \tilde r_t^\emptyset
        }{}
    \end{alignat*}
    
    The above and the fact that $\tilde r_t(\tau_1) = \tilde r_t^\emptyset$ for $\tau_1 > t$ makes \eqref{eq:full:31}:
    \begin{equation*}
        \forall [\tau_1, \tau_2] \sub [T] : \quad
        \sum_{t < \tau_1 } \tilde r_t(\emptyset)
        +
        \sum_{t \in [\tau_1, \tau_2]} \tilde r_t(\tau_1)
        -
        \sum_{t \in [\tau_2]} \tilde r_t(\emptyset)
        \le
        4 U_{\tau_2} \sqrt{T \log T}
    \end{equation*}
    or equivalently
    \begin{equation*}
        \forall [\tau_1, \tau_2] \sub [T] : \quad
        \sum_{t \in [\tau_1, \tau_2]} \tilde r_t(\tau_1)
        -
        \sum_{t \in [\tau_1, \tau_2]} \tilde r_t(\emptyset)
        \le
        4 U_{\tau_2} \sqrt{T \log T}
        .
    \end{equation*}
    
    Given the regret bound of each algorithm $\calA_{\tau_1}$ by round $\tau_2$ the above implies
    \begin{equation*}
        \forall [\tau_1, \tau_2] \sub [T] : \quad
        \sup_a \sum_{t \in [\tau_1, \tau_2]} r_t(a)
        -
        \sum_{t \in [\tau_1, \tau_2]} \tilde r_t(\emptyset)
        \le
        4 U_{\tau_2} \sqrt{T \log T}
        +
        \reg_{\tau_1}(\tau_2)
        .
    \end{equation*}
    which is the desired regret bound.
\end{proof}

Using a simple concentration inequality and the union bound, we prove \cref{thm:full:primal} from \cref{cor:full:primal_expectation}.

\begin{proof}[Proof of \cref{thm:full:primal}]
    Fix $1 \le \tau_1 < \tau_2 \le T$.
    For $t\in[\tau_1, \tau_2]$, define $X_t = r_t(b_t) - \sum_b q_t(b) r_t(b)$ and $M_t = \sum_{t'\in[\tau_1, t]} X_t$.
    We notice that the sequence $M_t$ is a martingale with respect to the history of the previous rounds $\mathcal H_{t-1}$ (player's and adversary's decisions): for every $t$
    \begin{equation*}
        \ExC{M_t - M{t-1}}{\mathcal H_{t-1}}
        =
        \ExC{X_t}{\mathcal H_{t-1}}
        =
        0
    \end{equation*}
    
    In addition we notice that $|X_t| \le U_t$ since $r_t(b) \in [0, U_t]$.
    This allows us to use Azuma's inequality, proving that for every $\d \in [0,1]$, with probability at least $1 - \d$ it holds
    \begin{equation*}
        M_{\tau_2}
        \ge
        - U_{\tau_2} \sqrt{2 (\tau_2 - \tau_1) \log(1/\d)}
    \end{equation*}
    which implies that with probability at least $1 - \d$
    \begin{equation*}
        \sum_{t = \tau_1}^{\tau_2} r_t(b_t)
        -
        \sum_{t = \tau_1}^{\tau_2} \sum_b q_t(b) r_t(b)
        \ge
        - U_{\tau_2} \sqrt{2 T \log(1/\d)}
    \end{equation*}
    
    Using the union bound over all $1 \le \tau_1 < \tau_2 \le T$ we get that for every $\d \in [0,1]$ with probability at least $1 - \d$ it holds that for all $1 \le \tau_1 < \tau_2 \le T$
    \begin{equation*}
        \sum_{t = \tau_1}^{\tau_2} r_t(b_t)
        -
        \sum_{t = \tau_1}^{\tau_2} \sum_b q_t(b) r_t(b)
        \ge
        - U_{\tau_2} \sqrt{2 T \log\frac{\binom{T}{2}}{\d}}
    \end{equation*}
    
    Using \cref{cor:full:primal_expectation} we get the theorem.
\end{proof}

%% file: appendix/4.tight.tex
\section{Deferred proofs of Section \ref{sec:tight_roi}} \label{sec:app:tight}

We first prove the reduction of how to turn approximate ROI satisfaction to an exact one.

\begin{proof}[Proof of \cref{lem:tight:reduction}]
    We first notice that since $\calA_2$ never has value less than payment and $\calA_1$ is run only when the accumulated value is at least $1$ higher than the payment, the ROI constraint is never going to be violated.
    Now we need to prove the total value guarantee.
    
    Assume that the high probability bounds of the two algorithms are true (which happens with probability at least $1 - 2\d$ for any $\d$ due to the union bound).
    Let $\tau$ be the last round when algorithm $\calA_2$ is run.
    Let $\mathcal T_1$ be the rounds up to $\tau$ where algorithm $\calA_1$ is run and $\mathcal T_2$ be the rounds up to $\tau$ where algorithm $\calA_2$ is run; note that $|\mathcal T_2|$ is the total number of rounds $\calA_2$ is run in total.
    We now have
    \begin{equation*}
        \sum_{t \in [\tau]} \One{b_t \ge d_t}(v_t - p(b_t, d_t))
        \ge
        Q(|\calT_2|, \d)
        -
        V(|\calT_1|, \d)
    \end{equation*}
    where the inequality follows from the ROI guarantees of the two algorithms.
    Using the fact that on round $\tau$ we run $\calA_2$ which upper bounds the above quantity by $2$ we get
    \begin{equation*}
        Q(|\calT_2|, \d)
        \le
        2 + V(|\calT_1|, \d)
        \le
        2 + V(T, \d)
    \end{equation*}
    
    Using the definition that $Q^{-1}(V(T, \d), \d)$ is the solution to $Q(x, \d) = V(T, \d)$ we get
    \begin{equation*}
        |\calT_2|
        \le
        Q^{-1}(V(T, \d), \d)
    \end{equation*}
    
    This proves that the total number of rounds $\calA_1$ was run is at least $T - Q^{-1}(1 + V(T))$.
    This means that the total regret is at most
    \begin{equation*}
        \reg_1\qty( T - Q^{-1}(1 + V(T)) )
        +
        2 Q^{-1}(1 + V(T))
    \end{equation*}
    where the second term represents the rounds $\calA_2$ was run instead of $\calA_1$ and the loss because of the budget consumption of $\calA_2$, which is at most $Q^{-1}(V(T, \d), \d)$ making the overall algorithm run out of budget $Q^{-1}(V(T, \d), \d)$ rounds earlier, missing out on that much utility.
\end{proof}

We now prove \cref{thm:tight:main_tight}.

\begin{proof}[Proof of \cref{thm:tight:main_tight}]
    As we explained in \cref{sec:tight_roi} there are algorithms that satisfy the assumptions of \cref{lem:tight:reduction} with
    \begin{equation*}
        \reg(T, \d)
        =
        V(T, \d)
        =
        \frac{1}{\rho \b} \, \order{
            \sqrt{L T} \log T
            +
            \sqrt{T \log(1/\d)}
        }
    \end{equation*}
    and
    \begin{equation*}
        Q(T_2, \d)
        =
        \b T_2
        -
        \order{
            \sqrt{L T_2} \log T_2
            +
            \sqrt{T_2 \log(T_2/\d)}
        }
        =
        \Omega \qty( \b T_2 )
    \end{equation*}
    where the last inequality follows from
    \begin{equation*}
        \b
        =
        \omega \qty(
            \sqrt{\frac{L}{T}} \log T
            +
            \sqrt{\frac{\log(T/\d)}{T}}
        )
    \end{equation*}
    
    The theorem follows by algebraic calculations.
\end{proof}

We now prove \cref{thm:tight:beta}.

\begin{proof}[Proof of \cref{thm:tight:beta}]
    Each round the value/highest-competing-bid distribution is the following
    \begin{equation*}
        (v_t, d_t) =
        \begin{cases}
            (1, 0), & \text{ w.p. } \b \\
            \left( \frac{1 - 2\b}{1 - \b}, 1 \right), & \text{ w.p. } 1-\b
        \end{cases}
    \end{equation*}
    
    It is not hard to see that the LP optimum bids $1$ every round and on expectation has reward $\opt = 1 - \b$ and ROI violation of $0$.
    
    Let $R_t$ be the cumulative ``ROI amount'' that any algorithm has collected in round $t$, i.e. the total value minus the total price paid.
    Because the algorithm needs to satisfy the ROI constraint with probability $1$, for all rounds it must hold that $R_t \ge 0$.
    This means that the best any algorithm can do is to bid $1$ in round $t$ if $R_{t-1} \ge \frac{\b}{1 - \b}$ (thus guaranteeing to always win) or bid less than $1$ if $R_{t-1} < \frac{\b}{1 - \b}$ (thus guaranteeing to win only if $d_t = 0$).
    
    Let $N_t$ be the number of rounds the algorithm bid less than $1$ and missed an item with $d_t = 1$.
    The regret of the algorithm is $N_T \frac{1 - 2\b}{1 - \b}$ so we need to calculate $\Ex{N_T}$.
    
    Let $R_t'$ be the cumulative ``ROI amount'' if the algorithm won every item. We notice that
    \begin{equation*}
        R_t - R_t'
        =
        \frac{\b}{1 - \b} N_t
    \end{equation*}
    since the difference does not change from round to round if the algorithm wins the item but increases if the the algorithms misses a high value item. Since $R_t \ge 0$ and it holds that
    \begin{equation*}
        N_T
        =
        \min_t N_t
        \ge
        - \frac{1 - \b}{\b} \min_t R_t'
    \end{equation*}

    We now note that $R_t'$ is an unbiased random walk, since
    \begin{equation*}
        R_t' - R_{t-1}'
        =
        \begin{cases}
            1, & \text{ w.p. } \b \\
            \frac{\b}{1 - \b}, & \text{ w.p. } 1-\b
        \end{cases}
    \end{equation*}
    and $\E[R_t' - R_{t-1}'] = 0$ and $\text{Var}[R_t' - R_{t-1}'] = \frac{\b}{1 - \b}$.
    We will show that $\Ex{\min_t R_t'} \le - \sqrt{T \frac{\b}{2\pi(1-\b)}}$, which in turn implies that
    \begin{equation*}
        \Ex{N_T}
        \ge
        \frac{1-\b}{\b}
        \sqrt{T \frac{\b}{2\pi(1-\b)}}
        =
        \sqrt{T \frac{1-\b}{2\pi\b}}
    \end{equation*}
    which implies that the expected regret of the algorithm wrt to the LP optimum is
    \begin{equation*}
        (1 - 2\b)
        \sqrt{T \frac{1}{2\pi\b(1-\b)}}
    \end{equation*}
    which proves the theorem.
    
    We now prove that $\Ex{\min_t R_t'} \le - \sqrt{T \frac{\b}{2\pi(1-\b)}}s$.
    We have that
    \begin{alignat*}{3}
        \Line{
            \Ex{\min_{t = 0, \ldots, T} R_t'}
        }{=}{
            \Ex{\min_{t = 0, \ldots, T} \min\left\{ 0, R_t' \right\}}
        }{}
        \\
        \Line{}{\le}{
            \Ex{ \min\left\{ 0, R_T' \right\}}
        }{}
        \\
        \Line{}{=}{
            - \frac{1}{2}\Ex{ \abs{ R_T' }}
        }{\substack{
            \Ex{ R_T' } = \Ex{ \max\{0 , R_T'\} } + \Ex{ \min\{0 , R_T'\} } = 0
            \\
            \Ex{\abs{R_T'}} = \Ex{ \max\{0 , R_T'\} } - \Ex{ \min\{0 , R_T'\} }
        }}
        \\
        \Line{}{\approx}{
            - \frac{1}{2} \Ex{\abs{
                G\qty(0, \sqrt{T \frac{\b}{1-\b}})
            }}
        }{\text{Central Limit Theorem for large } T}
        \\
        \Line{}{=}{
            - \frac{1}{2} \sqrt{T \frac{\b}{1-\b}} \sqrt{ \frac{2}{\pi} }
        }{}
    \end{alignat*}
    where the last equality uses standard facts of $G\qty\big(0, \sqrt{T \frac{\b}{1-\b}})$, the $0$-mean Gaussian with standard deviation $\sqrt{T \frac{\b}{1-\b}}$.
    
    Instead of using the central limit theorem one could explicitly use the fact that $R_T' = M \qty\big( 1 + \frac{\b}{1 - \b} ) - T \frac{\b}{1 - \b}$ where $M$ is a binomial random variable with $T$ tries and probability of success $b$ to calculate
    \begin{alignat*}{3}
        \Line{
            \Ex{ \abs{ R_T' }}
        }{=}{
            \qty( 1 + \frac{\b}{1 - \b} )\Ex{ \abs{ M - T \b }}
        }{}
        \\
        \Line{}{=}{
            \qty( 1 + \frac{\b}{1 - \b} )
            2 \lceil T \b \rceil
            (1 - \b)^{T + 1 - \lceil T \b \rceil}
            \b ^{\lceil T \b \rceil}
            \binom{T}{\lceil T \b \rceil}
        }{}
        \\
        \Line{}{=}{
            \sqrt{\frac{2 \b}{\pi(1-\b)}}\sqrt T
            -
            \order{\frac{1}{\sqrt T}}
        }{}
    \end{alignat*}
    which leads to a similar bound.
\end{proof}

%% file: appendix/5.bandit.tex
\section{Deferred proof of Section \ref{sec:bandit}} \label{sec:app:bandit}

\subsection{Deferred proof of Section \ref{ssec:bandit:lower}} \label{sec:app:bandit:lower}

Here we include the proof of \cref{thm:bandit_lower_bound}. 

\begin{proof}[Proof of \cref{thm:bandit_lower_bound}]
    Fix the player's budget, $\rho = 1/4$.
    Define the following
    \begin{itemize}
        \item $K = 2 T^{1/3}$ for some positive $c$ that will be defined later.
        \item $\e = \frac{1}{3K}$.
        \item For $i = 0, 1, 2, \ldots, K$ let $d_i = 1/3 + i\e$. Note $d_0 = 1/3$ and $d_K = 2/3$.
    \end{itemize}
    
    We will consider that $d_t$ can only be $0$, $1$, or one of the values in $\{d_i\}_{i = 0, \ldots, K-1}$; note that even though we defined the value $d_K$, $d_t$ cannot take that value.
    We will consider different distributions that can generate $d_t$, each specified by a CDF.
    We first define the base CDF-like distribution:
    \begin{equation}\label{eq:imp:cdf_base}
        F_\btt(x)
        =
        \frac{3}{4 - x}
    \end{equation}
    
    We use $\PRs_\btt$ and $\Es_\btt$ to denote the probability and expectation when $d_t$ is generated by $F_\btt$, which means that for any $x = 0, d_0, d_1, \ldots, d_{K-1}$, $\Prs[\btt]{d_t \le x} = F_\btt(x)$ and $\Prs[\btt]{d_t \le 1} = 1$.
    More precisely, even though we do not use the description of the probability density function, we have
    \begin{equation*}
        \Prs[\btt]{d_t = x} = 
        \begin{cases}
            F_\btt(0) = 3/4
            &, \text{ if } x = 0
            \\
            F_\btt(d_0) - F_\btt(0) = 3/44
            &, \text{ if } x = d_0 = 1/3
            \\
            F_\btt(d_i) - F_\btt(d_{i-1}) = \Theta(\e)
            &, \text{ if } x = d_i \text{ for some } i = 1, 2, \ldots, K-1
            \\
            1 - F_\btt(d_{K-1})% = 1 - \frac{3}{4}\frac{1}{1 - \rho(1/2 + \frac{K-1}{4K})}
            &, \text{ if } x = 1
            \\
            0 
            &, \text{ otherwise}
        \end{cases}
        .
    \end{equation*}
    
    We note that $F_\btt(d_{K-1}) < F_\btt(d_K) = 1/10$; this guarantees that our distribution is well defined.
    
    Given the base CDF $F_\btt$ we now define a different CDF-like function, $F_j$ for every $j = 0, 1, \ldots, K-1$.
    The distribution of $d_t$ associated with $F_j$ is going to be identical to the one associated with $F_\btt$. except for the probability of 
    \begin{equation}\label{eq:imp:cdf_d}
        F_j(x)
        =
        \begin{cases}
            F_\btt(x), &\text{ if } x \notin [d_j, d_{j+1})
            \\
            F_\btt(d_{j+1}), &\text{ if } x \in [d_j, d_{j+1})
        \end{cases}
    \end{equation}
    
    We use $\PRs_j$ and $\Es_j$ to denote the probability and expectation when $d_t$ is generated by $F_j$. This means that $\Prs[j]{d_t \le x} = \Prs[\btt]{d_t \le x}$ for all $x = 0, d_0, d_1, \ldots, d_{j-1}, d_{j+1}, \ldots, d_{K-1}, 1$ and $\Prs[j]{d_t \le d_j} = \Prs[\btt]{d_t \le d_{j+1}}$.
    
    We are going to assume that before round $t$, $j$ is picked adversarially.
    We now prove what a lower bound on the value of the optimal solution under $F_j$.
    
    \begin{lemma}\label{cl:imp:jth}
        When $d_t$ is generated according to $F_j$ for any $j=0, \ldots, K-1$ as explained above, then the value of the optimal solution is
        \begin{equation*}
            \opt_j
            \ge
            \frac{13}{16} T
            +
            \frac{3}{32} T \e
            % - O(\sqrt T)
        \end{equation*}
    \end{lemma}
    
    \begin{proof}
        We consider the strategy that bids $d_j$ with probability $\frac{11 - 3 (i + 1) \e}{12 (1 + 3 i \e)}$ and bids $0$ otherwise.
        The expected per round payment of this strategy is
        \begin{equation*}
            \frac{11 - 3 (i + 1) \e}{12 (1 + 3 i \e)} d_j F_j(d_j)
            = 
            \frac{11 - 3 (i + 1) \e}{12 (1 + 3 i \e)} \qty( \frac{1}{3} + i \e ) \frac{3}{4 - (\frac{1}{3} + (i+1) \e) }
            =
            \frac{1}{4}
        \end{equation*}
        which means it satisfies the budget constraint in expectation, since $\rho = 1/4$.
        The expected per-round value of this solution is
        \begin{alignat*}{3}
            & \frac{11 - 3 (i + 1) \e}{12 (1 + 3 i \e)} F_j(d_j)
            +
            \qty( 1 - \frac{11 - 3 (i + 1) \e}{12 (1 + 3 i \e)} ) F_j(0)
            \\
            & \quad =
            \frac{11 - 3 (i + 1) \e}{12 (1 + 3 i \e)} \frac{3}{4 - (\frac{1}{3} + (i+1) \e) }
            +
            \qty( 1 - \frac{11 - 3 (i + 1) \e}{12 (1 + 3 i \e)} ) \frac{3}{4}
            \\
            & \quad =
            \frac{13}{16}
            +
            \frac{3 \e}{16(1 + 3 i \e)}
            \ge
            \frac{13}{16}
            +
            \frac{3 \e}{16(1 + 3 K \e)}
        \end{alignat*}
        where the last inequality follows from $i \le K$.
        Substituting $\e = \frac{1}{3K}$ we get the lemma.
    \end{proof}

    In round $t$ the player bids $b_t$ and observes $x_t \in \{0, 1\}$ (if she won or not).
    We are going to assume that the player runs some deterministic algorithm, which comes w.l.o.g. since the environment is randomized.
    This means that the player's bid $b_t$ in round $t$ is a deterministic function of $x_1, \ldots, x_{t-1}$.
    We denote $x = x_{1:T}$.
    We also assume w.l.o.g. that the player's bids are always one of $0, d_0, \ldots, d_{K-1}, 1$, since any other bid is suboptimal.
    Let $N_i$ be the total number of times the player bids $d_i$ and $M_i$ be the total number of times the player bids $d_i$ and wins.
    We prove the following.
    
    \begin{lemma}\label{cl:imp:close}
        Let $g:\{0, 1\}^T \to \R$ be a function defined on $x$. Then, for every $j$, it holds that
        \begin{equation*}
            \left| \Es_j g(x) - \Es_\btt g(x)  \right|
            \le
            \left( \max_x |g(x)| \right) \e \sqrt{ 2 \Exs[\btt]{M_j} }
        \end{equation*}
    \end{lemma}
    
    \begin{proof}
        For any $j$ and $i$ we have
        \begin{alignat*}{3}
            \Line{
                \left| \Es_j f(x) - \Es_\btt f(x)  \right|
            }{=}{
                \left|
                    \sum_{x\in\{0,1\}^T} f(x) \left( \Prs[j]{x} - \Prs[\btt]{x} \right)
                \right|
            }{}
            \\
            \Line{}{\le}{
                \left( \max_x |f(x)| \right)
                \sum_{x\in\{0,1\}^T} \left|
                    \Prs[j]{x} - \Prs[\btt]{x}
                \right|
            }{}
            \\
            \Line{}{=}{
                \left( \max_x |f(x)| \right)
                \left| \mu_j - \mu_\btt \right|_1
            }{}
        \end{alignat*}
        where $\mu_j$ and $\mu_\btt$ are the probability distributions on $x$ given that $d_t$ is sampled from $F_j$ and $F_\btt$, respectively.
        Using the properties of the KL divergence, we have that
        \begin{equation*}
            \left| \mu_j - \mu_\btt \right|_1^2
            \le
            2 \KL{\mu_\btt}{\mu_j}
        \end{equation*}
        
        Now fix some sequence $x$. Conditioned on $x_{1:t-1}$, $b_t$ is deterministic. In addition, if $b_t \ne d_j$ it holds
        \begin{equation*}
            \mu_\btt(x_t = 1 | x_{1:t-1})
            =
            \mu_j(x_t = 1 | x_{1:t-1})
        \end{equation*}
        which, if $b_t \ne d_j$, implies
        \begin{equation*}
            \KL{\mu_\btt(x_t = \cdot | x_{1:t-1})}{\mu_j(x_t = \cdot | x_{1:t-1})}
            =
            0
        \end{equation*}
        
        When $b_t = d_j$ it holds that
        \begin{align*}
            \mu_\btt(x_t = 1 | x_{1:t-1})
            =&\;
            F_\btt(d_j)
            =
            \frac{3}{4 - \left( d_0 + j \e \right)}
            \\
            \mu_j(x_t = 1 | x_{1:t-1})
            =&\;
            F_j(d_j)
            =
            F_\btt(d_{j+1})
            =
            \frac{3}{4 - \left( d_0 + (j+1) \e \right)}
        \end{align*}
        which implies that if $b_t = d_j$
        \begin{alignat*}{3}
            \Line{
                \KL{\mu_\btt(x_t = \cdot | x_{1:t-1})}{\mu_j(x_t = \cdot | x_{1:t-1})}
            }{=}{
                \KL{\text{Bern}(F_\btt(d_j))}{\text{Bern}(F_\btt(d_{j+1}))}
            }{}
            \\
            \Line{}{=}{
                F_\btt(d_j) \log \frac{F_\btt(d_j)}{F_\btt(d_{j+1})}
                +
                (1 - F_\btt(d_j)) \log \frac{1 - F_\btt(d_j)}{1 - F_\btt(d_{j+1})}
            }{}
            \\
            \Line{}{\le}{
                F_\btt(d_j) \left( \frac{F_\btt(d_j)}{F_\btt(d_{j+1})} - 1 \right)
                +
                (1 - F_\btt(d_j)) \left( \frac{1 - F_\btt(d_j)}{1 - F_\btt(d_{j+1})} - 1 \right)
            }{}
            \\
            \Line{}{=}{
                3\frac{\e^2}{(1 - d_j)(4 - d_{j+1})^2}
                \;\le\;
                \frac{81}{100} \e^2
            }{}
        \end{alignat*}
        where $\text{Bern}(p)$ is a Bernoulli random variable with mean $p$, the first inequality follows by using $\log x \le x - 1$ for all $x > 0$ and the last inequality follows by using $d_j,d_{j+1} \le 2/3$.
        This implies that for any $b_t$,
        \begin{equation*}
            \KL{\mu_\btt(x_t = \cdot | x_{1:t-1})}{\mu_j(x_t = \cdot | x_{1:t-1})}
            \le
            \frac{81}{100} \e^2 \One{b_t(x) = d_j}
        \end{equation*}
        
        Taking expectations over $x\sim \mu_\btt$ and adding over $t$ we get
        \begin{equation*}
            \KL{\mu_\btt}{\mu_j}
            \le
            \frac{81}{100} \e^2 \Exs[\btt]{N_j}
            \le
            \frac{99}{100} \e^2 \Exs[\btt]{M_j}
            \le
            \e^2 \Exs[\btt]{M_j}
        \end{equation*}
        where the last inequality follows from the fact that $\Exs[\btt]{M_j} = F_\btt(d_j) \Exs[\btt]{N_j}$ and $F_\btt(d_j) \ge 9/11$.
        Combining this with what we proved before we get
        \begin{equation*}
            \left| \Es_j f(x) - \Es_\btt f(x)  \right|
            \le
            \left( \max_x |f(x)| \right) \e \sqrt{ 2 \Exs[\btt]{M_j} }
        \end{equation*}
        which proves the claim.
    \end{proof}
    
    Now note that since the agent's payment is at least $d_j M_j$, it must always hold $d_j M_j \le T \rho$, implying $M_j \le T \rho / d_j \le \frac{3}{4} T$.
    Applying \cref{cl:imp:close} for $f(x) = M_j$ we get that for any $j$
    \begin{equation*}
        \Es_j M_j
        \le
        \Es_\btt M_j
        +
        \frac{3\sqrt 2}{4} T \e \sqrt{ \Es_\btt M_j }
    \end{equation*}
    
    We notice that because it must hold $\sum_j d_j M_j \le T \rho$, there must be some $j$ such that $\Es_\btt M_j \le \frac{3}{4K} T$.
    Fix that $j$ and use that in the above inequality, making the above inequality for that $j$
    \begin{equation*}
        \Es_j M_j
        \le
        \frac{3 T}{4 K}
        +
        \frac{3\sqrt 2}{4} T \e \sqrt{ \frac{3T}{4K} }
        =
        \frac{3}{4} \frac{T}{K}
        +
        \frac{3 \sqrt{6}}{8} \frac{T^{3/2}\e}{\sqrt K}
    \end{equation*}
    Recalling that $K = 2 T^{1/3}$ and $\e = 1 / (3K)$ we get
    \begin{equation*}
        \Es_j M_j
        \le
        \frac{3}{8} T^{2/3}
        +
        \frac{\sqrt{3}}{16} T
    \end{equation*}
    
    Since $\Es_j M_j = F_j(d_j) \Es_j N_j$ and $F_j(d_j) = \frac{3}{4 - d_{j+1}} \ge 9/11$ (which follows from $d_{j+1} \ge 1/3$) the above becomes
    \begin{equation}\label{eq:band:1}
        \Es_j N_j
        \le
        \frac{33}{56} T^{2/3}
        +
        \frac{11\sqrt 3}{144} T
        :=
        U T
    \end{equation}
    
    Using the above upper bound on $\Es_j N_j$, we are going to upper bound the value the player can earn in the $j$-th instance.
    To do that, we are going to consider that the player knows that she is playing against the distribution of $\Prs[j]{\cdot}$ but is restricted by \eqref{eq:band:1}.
    In addition, we consider that the player can satisfy the budget constraint in expectation, since this only increases the budget she can earn.
    Given this setting, the player's value is upper bounded by the following LP, where $n_i$ represents the expected number of times the player bids $d_i$, $\ell_0$ is the expected number of times the player bids $0$, and $\ell_1$ is the expected number of times the player bids $1$:
    \begin{align*}
        \max_{\ell_0, \ell_1, n_0, \ldots, n_{K-1}} \quad
        & F_j(0) \ell_0 + \ell_1 + \sum_{i=0}^{K-1} F_j(d_i) n_i
        \\
        & \ell_1 + \sum_{i=0}^{K-1} d_i F_j(d_i) n_i \le \frac{T}{4}
        \\
        & \ell_0 + \ell_1 + \sum_{i=0}^{K-1} n_i \le T
        \\
        & n_j \le U T
    \end{align*}
     
    We upper bound the value of the above using its dual:
    \begin{align*}
        \min_{\l, \mu, \nu} \quad
        & \l \frac{T}{4} + \mu T + \nu U T
        \\
        & \mu \ge F_j(0)
        \\
        & \l + \mu \ge 1
        \\
        & \l d_i F_j(d_i) + \mu \ge F_j(d_i), \quad\forall i\ne j
        \\
        & \l d_j F_j(d_j) + \mu + \nu \ge F_j(d_j)
    \end{align*}
    
    We notice that the solution $\l = 1/4$, $\mu = 3/4$, $\nu = \e /4$ is feasible (which follows from some simple but lengthy algebra).
    This means that the player's expected reward is at most 
    \begin{equation*}
        T\left(
            \frac{13}{16}
            +
            \frac{U\e}{4}
        \right)
    \end{equation*}
    
    Using \cref{cl:imp:jth} we get that the expected regret it at least
    \begin{alignat*}{3}
        \Line{\reg}{\ge}{
            \frac{3}{32} T \e
            -
            \frac{\e}{4}\left(
                \frac{33}{56} T^{2/3}
                +
                \frac{11\sqrt 3}{144} T
            \right)
        }{}
        \\
        \Line{}{=}{
            \frac{1}{64} T^{2/3}
            -
            \frac{11}{448} T^{1/3}
            -
            \frac{11\sqrt 3}{3456} T^{2/3}
        }{\e = \frac{T^{-1/3}}{6}}
        \\
        \Line{}{\approx}{
            0.01 T^{2/3} 
            -
            0.02 T^{1/3}
        }{}
    \end{alignat*}
    which is $\Omega(T^{2/3})$, as promised.
\end{proof}

\subsection{Deferred proofs of Section \ref{ssec:bandit:upper}} \label{ssec:app:bandit:upper}

We first present the algorithm for \cref{thm:bandit:gen_algo} which low interval regret.
The algorithm is a modification of the algorithm of \cite{DBLP:conf/nips/Neu15} so that it works for time-varying ranges.
We present the algorithm in \cref{algo:bandit}.

\begin{algorithm}[t]
% \SetAlgoNoLine
\DontPrintSemicolon
\caption{No interval regret algorithm for bandit information with for time-varying ranges}
\label{algo:bandit}
\KwIn{Number of rounds $T$, number of actions $K$}

Set $\s = \frac{1}{T}$, $\xi = \frac{1}{2 \sqrt{T K}}$, $\theta = \frac{1}{\sqrt{T K}}$\;

Initialize weight for each action $w_1(a) = 1 \quad \forall \, a\in [K]$\;
\For{$t \in [T]$}
{
    Calculate $p_t(a) = \frac{w_t(a)}{\sum_{a'} w_t(a)}$\;
    
    Sample and play $a_t \sim p_t(\cdot)$\;
    
    Receive $\ell_t(a_t)$\;
    
    Calculate $\tilde \ell_t(a) = \frac{\ell_t(a)}{p_t(a) + \xi} \One{a = a_t}$ for all $a \in [K]$\;
    
    Receive loss range $[0, U_{t+1}]$ and calculate $\eta_{t+1} = \frac{\theta}{U_{t+1}}$\;
    
    Calculate $w_{t+1}(a) = (1 - \sigma) w_t(a) \exp(- \eta_{t+1} \tilde\ell_t(a)) + \frac{\sigma}{K}\sum_{a'} w_t(a') \exp(- \eta_{t+1} \tilde\ell_t(a'))$\;
}
\end{algorithm}

We now prove \cref{thm:bandit:gen_algo}.

\begin{proof}[Proof of \cref{thm:bandit:gen_algo}]
    First we define
    \begin{equation*}
        p_{t+1}'(a)
        =
        \frac{
            p_t(a) \exp\qty( - \eta_t \tilde \ell_t(a) )
        }{
            \sum_{a'} p_t(a') \exp\qty( - \eta_t \tilde \ell_t(a') )
        }
    \end{equation*}
    
    Next we show that for every $a^*$
    \begin{equation}\label{eq:bandit:31}
        \sum_a p_t(a) \tilde \ell_t(a)
        -
        \tilde \ell_t(a^*)
        \le
        \frac{1}{\eta_t}
        \log\qty( \frac{p_{t+1}'(a^*)}{p_t(a^*)} )
        +
        \frac{\eta_t}{2} \sum_a p_t(a) \tilde \ell_t^2(a)
    \end{equation}
    
    First we notice that
    \begin{alignat*}{3}
        \Line{
            \sum_a p_t(a) \exp\qty( - \eta_t \tilde\ell_t(a) )
        }{\le}{
            1 - \eta_t \sum_a p_t(a) \tilde\ell_t(a) + \frac{\eta_t^2}{2} \sum_a p_t(a) \tilde\ell_t^2(a)
        }{\substack{x \ge 0 \implies\\e^{-x} \le 1 - x + x^2/2}}
        \\
        \Line{}{\le}{
            \exp\qty(
                - \eta_t \sum_a p_t(a) \tilde\ell_t(a)
                + \frac{\eta_t^2}{2} \sum_a p_t(a) \tilde\ell_t^2(a)
            )
            % 1 - \eta_t \sum_a p_t(a) \tilde\ell_t(a) + \frac{\eta_t^2}{2} \sum_a p_t(a) \tilde\ell_t^2(a)
        }{1 + x \le e^x}
        \\
    \end{alignat*}
    % 
    % where in the last inequality we used $\eta_t \tilde \ell_t(a) \le 1$ for all $a$.
    Using the definition of $p_{t+1}'(a^*)$ in the l.h.s. of the inequality above we get
    \begin{equation*}
        \exp\qty( - \eta_t \tilde \ell_t(a^*) )
        \frac{p_t(a^*)}{p_{t+1}'(a^*)}
        \le
        \exp\qty(
            - \eta_t \sum_a p_t(a) \tilde\ell_t(a)
            + \frac{\eta_t^2}{2} \sum_a p_t(a) \tilde\ell_t^2(a)
        )
    \end{equation*}
    
    Taking a logarithm and rearanging we get \eqref{eq:bandit:31}.
    
    We now notice that
    \begin{alignat*}{3}
        \Line{
            p_{t+1}(a^*)
        }{=}{
            \frac{
                w_{t+1}(a^*)
            }{
                \sum_a w_{t+1}(a)
            }
        }{}
        \\
        \Line{}{=}{
            \frac{
                (1 - \sigma) w_t(a^*) \exp(- \eta_t \tilde \ell_t(a^*)) + \frac{\sigma}{K}\sum_{a'} w_t(a') \exp(- \eta_t \tilde\ell_t(a'))
            }{
                \sum_a \qty(
                    (1 - \sigma) w_t(a) \exp(- \eta_t \tilde \ell_t(a)) + \frac{\sigma}{K}\sum_{a'} w_t(a') \exp(- \eta_t \tilde\ell_t(a'))
                )
            }
        }{}
        \\
        \Line{}{=}{
            \frac{
                (1 - \sigma) w_t(a^*) \exp(- \eta_t \tilde \ell_t(a^*)) + \frac{\sigma}{K}\sum_{a'} w_t(a') \exp(- \eta_t \tilde\ell_t(a'))
            }{
                \sum_a w_t(a) \exp(- \eta_t \tilde \ell_t(a))
            }
        }{}
        \\
        \Line{}{=}{
            \frac{
                (1 - \sigma) p_t(a^*) \exp(- \eta_t \tilde \ell_t(a^*)) + \frac{\sigma}{K}\sum_{a'} p_t(a') \exp(- \eta_t \tilde\ell_t(a'))
            }{
                \sum_a p_t(a) \exp(- \eta_t \tilde \ell_t(a))
            }
        }{}
        \\
        \Line{}{\ge}{
            (1 - \sigma) p_{t+1}'(a^*)  + 0
        }{}
    \end{alignat*}
    
    Bounding $p_{t+1}'(a^*)$ in \eqref{eq:bandit:31} we get
    \begin{equation*}
        \sum_a p_t(a) \tilde \ell_t(a)
        -
        \tilde \ell_t(a^*)
        \le
        \frac{1}{\eta_t}
        \log\qty( \frac{p_{t+1}(a^*)}{(1-\s) p_t(a^*)} )
        +
        \frac{\eta_t}{2} \sum_a p_t(a) \tilde \ell_t^2(a)
    \end{equation*}
    
    Let $I = [\tau_1, \tau_2]$; summing the above for all $i \in I$ we get
    \begin{equation}\label{eq:bandit:32}
        \sum_{t \in I}\sum_a p_t(a) \tilde \ell_t(a)
        -
        \sum_{t \in I}\tilde \ell_t(a^*)
        \le
        \sum_{t \in I}\frac{1}{\eta_t}
        \log\qty( \frac{p_{t+1}(a^*)}{(1-\s) p_t(a^*)} )
        +
        \sum_{t \in I}\frac{\eta_t}{2} \sum_a p_t(a) \tilde \ell_t^2(a)
    \end{equation}
    
    Now we focus on
    \begin{alignat}{3} \label{eq:bandit:33}
        \sum_{t \in I} & \frac{1}{\eta_t}
            \log\qty( \frac{p_{t+1}(a^*)}{(1-\s) p_t(a^*)} )
        \nonumber\\
        \Line{}{=}{
            \frac{\log p_{\tau_2 + 1}(a^*)}{\eta_{\tau_2}}
            -
            \frac{\log \qty( (1 - \s)p_{\tau_1 }(a^*) )}{\eta_{\tau_1}}
            % +
            % \sum_{t \in I \setminus \{\tau_1\}}
                % \log\qty( \frac{p_t^{1/\eta_{t-1}}(a^*)}{(1-\s)^{1/\eta_t} p_t^{1/\eta_t}(a^*)} )
        }{}
        \nonumber\\
        \Line{}{}{
            % \frac{\log p_{\tau_2 + 1}(a^*)}{\eta_{\tau_2}}
            % -
            % \frac{\log \qty( (1 - \s)p_{\tau_1 }(a^*) )}{\eta_{\tau_1}}
            \quad +
            \sum_{t \in I \setminus \{\tau_1\}}
                \log\qty( \frac{p_t^{1/\eta_{t-1}}(a^*)}{(1-\s)^{1/\eta_t} p_t^{1/\eta_t}(a^*)} )
        }{}
        \nonumber\\
        \Line{}{\le}{
            0
            -
            \frac{\log \qty( (1 - \s) \frac{\s}{K} )}{\eta_{\tau_1}}
            +
            \sum_{t \in I \setminus \{\tau_1\}}
                \log\qty( \frac{p_t^{\frac{1}{\eta_{t-1}} - \frac{1}{\eta_t}}(a^*)}{(1-\s)^{1/\eta_t}} )
        }{p_t(a) \in \qty[\frac{\s}{K}, 1]}
        \nonumber\\
        \Line{}{\le}{
            - \frac{\log \qty( (1 - \s) \frac{\s}{K} )}{\eta_{\tau_1}}
            +
            \sum_{t \in I \setminus \{\tau_1\}}
                \qty( \frac{1}{\eta_{t-1}} - \frac{1}{\eta_t} ) \log \frac{\s}{K}
        }{}
        \nonumber\\
        \Line{}{}{
            \quad +
            \sum_{t \in I \setminus \{\tau_1\}}
                \frac{1}{\eta_t}\log \frac{1}{1-\s}
        }{p_t(a) \ge \frac{\s}{K}}
        \nonumber\\
        \Line{}{=}{
            - U_{\tau_1} \frac{\log \qty( (1 - \s) \frac{\s}{K} )}{\theta}
            +
            \frac{1}{\theta} \log \frac{\s}{K} \sum_{t \in I \setminus \{\tau_1\}}
                \qty( U_{t-1} - U_t )
        }{}
        \nonumber\\
        \Line{}{}{
            \quad +
            U_{\tau_2}\frac{|I| - 1}{\theta}\log \frac{1}{1-\s}
        }{\eta_t = \frac{\theta}{U_t}}
        \nonumber\\
        \Line{}{\le}{
            - U_{\tau_2} \frac{\log \qty( (1 - \s) \frac{\s}{K} )}{\theta}
            -
            \frac{U_{\tau_2}}{\theta} \log \frac{\s}{K}
            +
            U_{\tau_2}\frac{|I| - 1}{\theta}\log \frac{1}{1-\s}
        }{}
        \nonumber\\
        \Line{}{=}{
            U_{\tau_2} \qty(
                \sqrt{T K} \log \frac{K T}{1 - \frac{1}{T}}
                +
                \sqrt{T K} \log (T K)
                +
                T^{3/2} \sqrt K \log \frac{1}{1 - \frac{1}{T}}
            )
        }{\substack{\s = \frac{1}{T}\\\theta = \frac{1}{\sqrt{K T}}}}
        \nonumber\\
        \Line{}{=}{
            U_{\tau_2} \sqrt{T K} \qty(
                \log \frac{K T^2}{T - 1}
                +
                \log (T K)
                +
                \sqrt{T} \log \frac{T}{T - 1}
            )
        }{}
        \nonumber\\
        \Line{}{\le}{
            U_{\tau_2} \sqrt{T K} \qty(
                2 \log \qty( K T^2 )
                +
                2 \frac{1}{\sqrt{T}}
            )
        }{\log \frac{T}{T - 1} \le \frac{2}{T}}
        \nonumber\\
        \Line{}{\le}{
            U_{\tau_2} \order{
                \sqrt{T K} \log(T K)
            }
        }{}
    \end{alignat}
    
    Now we bound
    \begin{equation}\label{eq:bandit:34}
        \sum_a p_t(a) \tilde \ell_t^2(a)
        =
        \sum_a p_t(a) \tilde \ell_t(a) \frac{\ell_t(a)}{p_t(a) + \xi}
        \le
        U_t \sum_a \tilde \ell_t(a)
    \end{equation}
    
    Using \eqref{eq:bandit:33} and \eqref{eq:bandit:34} in \eqref{eq:bandit:32}, and substituting $\eta_t = \theta/ U_t$ we get
    \begin{equation}\label{eq:bandit:39}
        \sum_{t \in I}\sum_a p_t(a) \tilde \ell_t(a)
        -
        \sum_{t \in I}\tilde \ell_t(a^*)
        \le
        U_{\tau_2} X
        +
        \frac{\theta}{2} \sum_{t \in I} \sum_a \tilde \ell_t(a)
    \end{equation}
    
    Using a slight modification of \cite[Corollary 1]{DBLP:conf/nips/Neu15} we get that
    \begin{equation}\label{eq:bandit:35}
        \Pr{
            \forall a \in [K]:\;
            \sum_{t \in I} \qty( \tilde \ell_t(a) - \ell_t(a) )
            \le
            U_{\tau_2}\frac{\log(K/\d)}{2 \xi}
        }
        \ge
        1 - \d
    \end{equation}
    
    Now we bound
    \begin{equation}\label{eq:bandit:36}
        \sum_a p_t(a) \tilde \ell_t(a)
        =
        p_t(a_t) \frac{\ell_t(a_t)}{p_t(a_t) + \xi}
        \ge
        \ell_t(a_t) - \ell_t(a_t) \frac{\xi}{p_t(a_t) + \xi}
        =
        \ell_t(a_t) - \xi \sum_a \tilde \ell_t(a)
    \end{equation}
    
    We now combine all the above to bound the regret. Assume that \eqref{eq:bandit:35} is true; for every $a^* \in [K]$:
    \begin{alignat*}{3}
        \sum_{t \in I} & \qty\big( \ell_t(a_t) - \ell_t(a^*) )
        \\
        \Line{}{\le}{
            \sum_{t \in I} \sum_a p_t(a) \tilde \ell_t(a)
            +
            \xi \sum_{t \in I} \sum_a \tilde \ell_t(a)
            -
            \sum_{t \in I} \ell_t(a^*)
        }{\text{by \eqref{eq:bandit:36}}}
        \\
        \Line{}{\le}{
            \sum_{t \in I} \sum_a p_t(a) \tilde \ell_t(a)
            +
            \xi \sum_{t \in I} \sum_a \tilde \ell_t(a)
            -
            \sum_{t \in I} \tilde \ell_t(a^*)
            +
            U_{\tau_2} \frac{\log(K/\d)}{2\xi}
        }{\text{by \eqref{eq:bandit:35}}}
        \\
        \Line{}{\le}{
            U_{\tau_2} \order{ \sqrt{T K} \log(T K) }
            +
            \frac{\theta}{2} \sum_{t \in I} \sum_a \tilde \ell_t(a)
            +
            \xi \sum_{t \in I} \sum_a \tilde \ell_t(a)
            +
            U_{\tau_2} \frac{\log(K/\d)}{2\xi}
        }{\text{by \eqref{eq:bandit:39}}}
        \\
        \Line{}{=}{
            U_{\tau_2} \order{ \sqrt{T K} \log(T K) }
            +
            U_{\tau_2} \frac{\log(K/\d)}{2\xi}
            +
            \qty( \frac{\theta}{2} + \xi ) \sum_{t \in I} \sum_a \tilde \ell_t(a)
        }{}
        \\
        \Line{}{\le}{
            U_{\tau_2} \order{ \sqrt{T K} \log(T K) }
            +
            U_{\tau_2} \frac{\log(K/\d)}{2\xi}
            +
            \qty( \frac{\theta}{2} + \xi ) K \qty( U_{\tau_2} |I| + U_{\tau_2}\frac{\log(K/\d)}{2 \xi} )
        }{\text{by \eqref{eq:bandit:35}}}
        \\
        \Line{}{=}{
            U_{\tau_2} \order{ \sqrt{T K} \log(T K) }
            +
            U_{\tau_2} \sqrt{T K}\log(K/\d)
            +
            \frac{1}{\sqrt{TK}} K \qty( U_{\tau_2} T + U_{\tau_2} \sqrt{T K} \log(K/\d) )
        }{\substack{\xi = \frac{1}{2 \sqrt{T K}}\\\theta = \frac{1}{\sqrt{T K}}}}
        \\
        \Line{}{=}{
            U_{\tau_2} \order{ \sqrt{T K} \log(T K / \d) }
            % +
            % U_{\tau_2} \sqrt{T K}\log(K/\d)
            +
            U_{\tau_2} \frac{\sqrt{K}}{\sqrt{T}} \qty( T + \sqrt{T K} \log(K/\d) )
        }{}
        \\
        \Line{}{=}{
            U_{\tau_2} \order{
                \sqrt{T K} \log(T K / \d) 
                +
                K \log(K/\d) 
            }
        }{}
    \end{alignat*}
    which proves the desired bound.
\end{proof}

We now proceed to prove \cref{thm:bandit_ub}.

\begin{proof}[Proof of \cref{thm:bandit_ub}]
    Fix $N = \lceil L^{3/4} T^{1/4} \rceil$ and $K = \lceil T^{1/4} / L^{1/4} \rceil$.
    Let $\tilde v_i = i/N$ for $i \in [N]$ and $\tilde b_j = j/N$ for $j \in [N]$.
    For every $i\in [N]$ define $r_t^i: [K] \to \R_{\ge 0}$ such that
    \begin{equation*}
        r_t^i(j) = \One{\tilde b_{t,j} \ge d_t} \qty(
            \chi_t \tilde v_i - \psi_t p(\tilde b_{t,j}, d_t)
        )
    \end{equation*}
    where $\tilde b_{t,j}$ is either $\tilde b_j$ or the safe bid of $\tilde r_t^i(\cdot)$.
    
    Define for every round $i_t = \lceil v_t N \rceil$ (i.e. the $i$ that corresponds to the value that is closest and above to $v_t$).
    Fix $i \in [N]$ and let $\calT_i \sub [T]$ be the rounds where $i = i_t$; let $T_i = |\calT_i|$.
    Let $\calA_i$ be an instance of \cref{algo:bandit} which is run only on rounds $\calT_i$.
    $\calA_i$ has $K$ actions.
    The reward of the $j$-th action in round $t$ is $\tilde r_t^i(j)$.
    Let $j_t$ be the output of $\calA_i$ in a round $t\in\calT_i$, which we use to bid $\tilde b_{t,j_t}$.
    Because of \cref{thm:bandit:gen_algo} we have that for every $\d > 0$ with probability at least $1 - \d$, for every $\tau_1 \le \tau_2$
    \begin{equation*}
        \max_{j \in [K]} \sum_{t \in \calT_i \cap [\tau_1, \tau_2]} \tilde r_t^i(j)
        -
        \sum_{t \in \calT_i \cap [\tau_1, \tau_2]} \tilde r_t^i(j_t)
        \le
        U_{\tau_2} \order{
            \qty( \sqrt{T_i K} + K ) \log \frac{T_i K}{\d}
        }
    \end{equation*}
    
    We do the above process for every $i$ and use it as an algorithm\footnote{For rounds where $v_t = 0$ we have not defined an algorithm; in those rounds we can bid $0$ (which is optimal) to guarantee no regret; for simplicity however we assume that $v_t > 0$.}.
    Doing this for every $i$ and using the union bound we get that for every $\d > 0$ with probability at least $1 - \d$, for every $i \in [N]$ and $\tau_1 \le \tau_2$
    \begin{alignat}{3} \label{eq:bandit:60}
        \Line{
            \max_{j \in [K]} \sum_{t \in \calT_i \cap [\tau_1, \tau_2]} \tilde r_t^{i_t}(j)
            -
            \sum_{t \in \calT_i \cap [\tau_1, \tau_2]} \tilde r_t^{i_t}(j_t)
        }{\le}{
            U_{\tau_2} \order{
                \qty( \sqrt{T_i K} + K ) \log \frac{T_i K N}{\d}
            }
        }{}
        \nonumber \\
        \Line{}{=}{
            U_{\tau_2} \order{
                \qty( \sqrt{T_i \frac{T^{1/4}}{L^{1/4}}} + \frac{T^{1/4}}{L^{1/4}} ) \log \frac{L T}{\d}
            }
        }{\substack{N = \lfloor L^{3/4} T^{1/4}\rfloor\\K =\lfloor T^{1/4} / L^{1/4}\rfloor}}
    \end{alignat}
    
    Fix $\tau_1, \tau_2$.
    We want to use \eqref{eq:bandit:60} to upper bound
    \begin{equation*}
        \sup_{f \in \calF} \sum_{t=\tau_1}^{\tau_2} r_t(f)
        -
        \sum_{t=\tau_1}^{\tau_2} r_t(\tilde b_{t, j_t})
    \end{equation*}
    
    Fix $f \in \calF$ and define $j_1, \ldots, j_N \in [K]$ such that $\tilde b_{j_i}$ (recall $\tilde b_j = j/K$) is the bid that is above $f(v)$ for every $v \in (v_{i-1}, v_i]$ and as small as possible, i.e., for all $i \in [N]$ it holds
    \begin{equation*}
        \tilde b_{j_i}
        =
        \left\lceil
            K \sup_{v \in (\tilde v_{i-1}, \tilde v_i]} f(v)
        \right\rceil / K
    \end{equation*}
    We notice that for all $i\in [N]$ and $v \in (v_{i-1}, v_i]$: $\tilde b_{j_i} \in [f(v), f(v) + \frac{L}{N} + \frac{1}{K}]$ (by $L$-Lipschitz of $f$).
    We now have
    \begin{alignat*}{3}
        \Line{
            r_t\qty\big(f)
        }{=}{
            \One{f(v_t) \ge d_t} \qty( \chi_t v_t - \psi_t p\qty\big( f(v_t) , d_t ) )
        }{}
        \\
        \Line{}{\le}{
            \One{\tilde b_{j_{i_t}} \ge d_t} \qty( \chi_t \tilde v_{i_t} - \psi_t p\qty( \tilde b_{j_{i_t}} - \frac{L}{N} - \frac{1}{K} , d_t ) )^+
        }{}
        \\
        \Line{}{\le}{
            \One{\tilde b_{j_{i_t}} \ge d_t} \qty( \chi_t \tilde v_{i_t} - \psi_t p\qty( \tilde b_{j_{i_t}}, d_t ) )^+
            +
            \psi_t \qty( \frac{L}{N} + \frac{1}{K} )
        }{}
        \\
        \Line{}{\le}{
            \tilde r_t^{i_t}(j_{i_t})
            +
            \psi_t \qty( \frac{L}{N} + \frac{1}{K} )
        }{}
    \end{alignat*}
    where the last inequality follows because the bid that corresponds to $j_{i_t}$ in round $t$ is the safe bid of that round whose reward is non-negative and at least as good as $\tilde b_{j_{i_t}}$'s.
    Summing the above over $t \in [\tau_1, \tau_2]$ and taking a sup over $f$ and a max over $j_1,\ldots, j_N$ we get
    \begin{alignat}{3} \label{eq:bandit:61}
        \Line{
            \sup_{f \in \calF} \sum_{t = \tau_1}^{\tau_2} r_t(f)
        }{\le}{
            \max_{j_1,\ldots,j_N} \sum_{t = \tau_1}^{\tau_2} \tilde r_t^{i_t}(j_{i_t})
            +
            U_{\tau_2} T \qty( \frac{L}{N} + \frac{1}{K} )
        }{}
        \nonumber\\
        \Line{}{=}{
            \sum_{i\in[N]} \max_{j \in [K]} \sum_{t \in \calT_i \cap [\tau_1, \tau_2]} \tilde r_t^{i_t}(j)
            +
            U_{\tau_2} T \qty( \frac{L}{N} + \frac{1}{K} )
        }{}
        \nonumber\\
        \Line{}{\le}{
            \sum_{i\in[N]} \max_{j \in [K]} \sum_{t \in \calT_i \cap [\tau_1, \tau_2]} \tilde r_t^{i_t}(j)
            +
            U_{\tau_2} 2 T^{3/4} L^{1/4}
        }{\substack{N \ge L^{3/4} T^{1/4}\\K \ge T^{1/4} / L^{1/4}}}
        \nonumber\\
        \Line{}{\le}{
            \sum_{t \in [\tau_1, \tau_2]} \tilde r_t^{i_t}(j_t)
            +
            \sum_{i \in [N]} U_{\tau_2} \order{
                \qty( \sqrt{T_i \frac{T^{1/4}}{L^{1/4}}} + \frac{T^{1/4}}{L^{1/4}} ) \log \frac{L T}{\d}
            }
            +
            U_{\tau_2} 2 T^{3/4} L^{1/4}
        }{\text{by \eqref{eq:bandit:60}}}
        \nonumber\\
        \Line{}{=}{
            \sum_{t \in [\tau_1, \tau_2]} \tilde r_t^{i_t}(j_t)
            +
            U_{\tau_2} \order{
                \qty(  L^{1/4} T^{3/4} + T^{1/2} L^{1/2} ) \log \frac{L T}{\d}
                +
                T^{3/4} L^{1/4}
            }
        }{N = \lfloor L^{3/4} T^{1/4}\rfloor}
        \nonumber\\
        \Line{}{=}{
            \sum_{t \in [\tau_1, \tau_2]} \tilde r_t^{i_t}(j_t)
            +
            U_{\tau_2} \order{
                \qty(  L^{1/4} T^{3/4} + T^{1/2} L^{1/2} ) \log \frac{L T}{\d}
            }
        }{}
    \end{alignat}
    where the second to last inequality holds by the fact that $\sum_i \sqrt{T_i} \le \sqrt{N T}$ since $\sum_i T_i = T$.
    By noticing that $\tilde r_t^i(j_t) \le r_t(\tilde b_{t, j_t}) + U_t(1/K + 1/N)$ and bounding $T(1/N + 1/K) = O(L^{1/4} T^{3/4})$ we get the high probability interval regret bound on the rewards $r_t(\cdot)$.
    Using \cref{thm:pre:primal_dual} we get the desired result.
    
    The tight satisfaction of the ROI constraint follows by using \cref{lem:tight:reduction}.
    The first algorithm is the one we describe above.
    The second algorithm is the primal algorithm described above with $\chi_t = \psi_t = 1$, which achieves $Q(\tau, \d) = \tau \b - \order*{ (  L^{1/4} \tau^{3/4} + \tau^{1/2} L^{1/2} ) \log \frac{L T}{\d} }$.
\end{proof}

%% file: appendix/6.poly.tex
\section{Deferred proofs of Section \ref{sec:poly}} \label{sec:app:poly}

% \begin{algorithm}[t]
% % \SetAlgoNoLine
% \DontPrintSemicolon
% \caption{Full information polynomial time algo for auctions}
% \label{algo:poly}
% \KwIn{Number of rounds $T$, Lipschitz parameter $L$}

% Set $N = L^{2/3} T^{1/3}$ and $K = T^{1/3}$\;

% Let $\tilde v_i = \frac{i}{N}$ for $i \in [N]$
% \tcp*{Discretization of values}

% Let $\tilde b_i = \frac{j}{N}$ for $j \in [K]$
% \tcp*[f]{Discretization of bids}

% Initialize algorithm $\calA_i$ \cref{algo:full:good}, for $i \in [N]$\;

% \For{$t \in [T]$}
% {
%     Receive $v_t, \chi_t, \psi_t$ and set $U_t = \max_{t' \le t}\{\chi_{t'}, \psi_{t'}\}$\;
    
%     Let $\tilde b_{t,j} = \text{safe}(\tilde b_i)$\;
    
%     Pass $U_t$ into each $\calA_i$\;
    
%     Let $i_t = \lceil v_t N \rceil/N$\;
    
%     Receive distribution $p_t^{i_t}(\cdot)$ from $\calA_{i_t}$\;
    
%     Bid $\tilde b_{t,j_t}$ where $j_t \sim p_t^{i_t}(\cdot)$\;
    
%     Receive $d_t$\;
    
%     Set $\tilde r_t^i(j) = \One{\tilde b_{t,j} \ge d_t} \qty\big( \chi_t \tilde v_i - \psi_t p(\tilde b_{t,j}, d_t) )$ for $i \in [N]$ and $j \in [K]$\;
    
%     Pass each $\tilde r_t^i(\cdot)$ to each $\calA_i$\;
% }
% \end{algorithm}

\begin{proof}[Proof of \cref{thm:poly}]
    We use a scheme similar to the one used in the proof of \cref{thm:bandit_ub}.
    Fix $N = \lfloor L T \rfloor$ and $K = T$.
    Let $\tilde v_i = i/N$ for $i \in [N]$ and $\tilde b_j = j/N$ for $j \in [N]$.
    For every $i\in [N]$ define $r_t^i: [K] \to \R_{\ge 0}$ such that
    \begin{equation*}
        r_t^i(j) = \One{\tilde b_{t,j} \ge d_t} \qty(
            \chi_t \tilde v_i - \psi_t p(\tilde b_{t,j}, d_t)
        )
    \end{equation*}
    where $\tilde b_{t,j}$ is either $\tilde b_j$ or the safe bid (\cref{asmp:full}) of $\tilde r_t^i(\cdot)$.
    
    Define for every round $i_t = \lceil v_t N \rceil$ (i.e. the $i$ that corresponds to the value that is closest and above to $v_t$).
    For every $i\in[N]$ let $\calA_i$ be an instance of the algorithm in \cref{thm:full:good1} for $\Delta = 1$ with $K$ actions that has also been modified to have low interval regret, using \cref{thm:full:interval}.
    Unlike the bandit setting, here we run $\calA_i$ even in rounds where $i_t \ne i$.
    In every round $\calA_i$ outputs an action $j_t^i \in [K]$ which corresponds to the bid $\tilde b_{t,j_t^i}$.
    The reward of the $j$-th action in round $t$ is $\tilde r_t^i(j)$.
    The bid we use in round $t$ is the one suggested by algorithm $\calA_{i_t}$, $\tilde b_{t,j_t^{i_t}}$.
    
    Using \cref{thm:full:good1,thm:full:interval}, an application of Azuma's inequality and a union bound over $i$, we have that with probability at least $1 - \d$ for every $1 \le \tau_1 < \tau_2 \le T$ and every $i \in [N]$:
    \begin{equation}\label{eq:poly:60}
        \max_{j \in [K]} \sum_{t=\tau_1}^{\tau_2} \tilde r_t^i(j)
        -
        \sum_{t=\tau_1}^{\tau_2} \tilde r_t^i(j_t^i)
        \le
        U_{\tau_2} \order{ \sqrt{T \log (T K N / \d)} }
    \end{equation}
    
    Fix $\tau_1, \tau_2$.
    To get our regret bound we have to upper bound 
    \begin{equation*}
        \sup_{f \in \calF} \sum_{t=\tau_1}^{\tau_2} r_t(f)
        -
        \sum_{t=\tau_1}^{\tau_2} \sum_i \One{i_t = i} r_t(\tilde b_{t, j_t^i})
    \end{equation*}
    
    Fix $f \in \calF$ and define $j_1, \ldots, j_N \in [K]$ such that for all $i \in [N]$
    \begin{equation*}
        \tilde b_{j_i}
        =
        \left\lceil
            K \sup_{v \in (\tilde v_{i-1}, \tilde v_i]} f(v)
        \right\rceil / K
    \end{equation*}
    We notice that for all $i\in [N]$ and $v \in (v_{i-1}, v_i]$: $\tilde b_{j_i} \in [f(v), f(v) + \frac{L}{N} + \frac{1}{K}]$ (by $L$-Lipschitz of $f$).
    We now have
    \begin{alignat*}{3}
        \Line{
            r_t\qty\big(f)
        }{=}{
            \One{f(v_t) \ge d_t} \qty( \chi_t v_t - \psi_t p\qty\big( f(v_t) , d_t ) )
        }{}
        \\
        \Line{}{\le}{
            \One{\tilde b_{j_{i_t}} \ge d_t} \qty( \chi_t \tilde v_{i_t} - \psi_t p\qty( \tilde b_{j_{i_t}} - \frac{L}{N} - \frac{1}{K} , d_t ) )^+
        }{}
        \\
        \Line{}{\le}{
            \One{\tilde b_{j_{i_t}} \ge d_t} \qty( \chi_t \tilde v_{i_t} - \psi_t p\qty( \tilde b_{j_{i_t}}, d_t ) )^+
            +
            \psi_t \qty( \frac{L}{N} + \frac{1}{K} )
        }{p(\cdot,d): 1-\text{Lipschitz}}
        \\
        \Line{}{\le}{
            \tilde r_t^{i_t}(j_{i_t})
            +
            \psi_t \qty( \frac{L}{N} + \frac{1}{K} )
        }{}
    \end{alignat*}
    where the last inequality follows because the bid that corresponds to $j_{i_t}$ in round $t$ is the safe bid of that round whose reward is non-negative and at least as good as $\tilde b_{j_{i_t}}$'s.
    Summing the above over $t \in [\tau_1, \tau_2]$ and taking a sup over $f$ and a max over $j_1,\ldots, j_N$ we get
    \begin{alignat}{3} \label{eq:poly:61}
        \Line{
            \sup_{f \in \calF} \sum_{t = \tau_1}^{\tau_2} r_t(f)
        }{\le}{
            \max_{j_1,\ldots,j_N} \sum_{t = \tau_1}^{\tau_2} \tilde r_t^{i_t}(j_{i_t})
            +
            U_{\tau_2} T \qty( \frac{L}{N} + \frac{1}{K} )
        }{}
        \nonumber\\
        \Line{}{=}{
            \sum_{i\in[N]} \max_{j^* \in [K]} \sum_{t = \tau_1}^{\tau_2} \One{i_t = i} \tilde r_t^i(j^*)
            % \sum_{i\in[N]} \One{i_t = i} \max_{j \in [K]} \sum_{t \in \calT_i \cap [\tau_1, \tau_2]} \tilde r_t^{i_t}(j)
            +
            2 U_{\tau_2}
        }{\substack{N \ge \lfloor T L \rfloor\\K = T}}
    \end{alignat}
    
    We now prove a high probability bound on $\sum_{t = \tau_1}^{\tau_2} \One{i_t = i} \tilde r_t^i(j^*)$, for a fixed $j^*$.
    Using Azuma's inequality we can prove that for all $\d > 0$ with probability at least $1 - \d$,
    \begin{equation*}
        \sum_{t = \tau_1}^{\tau_2} \One{i_t = i} \tilde r_t^i(j^*)
        \le
        \sum_{t = \tau_1}^{\tau_2} \Pr{i_t = i} \tilde r_t^i(j^*)
        +
        \order{ U_{\tau_2} \sqrt{T \log(1/\d)} }
    \end{equation*}
    Note that to use Azuma's inequality we have to rely on the fact that $i_t$ and $\tilde r_t^i(j^*)$ are independent conditioned on the history of rounds up to $t-1$, since $\tilde r_t^i(j^*)$ does not depend on $v_t$ but only $d_t$
    
    By letting $\Pr{i_t = i} = q_i$ and taking a union bound over all $j* \in [K]$ and all $i \in [N]$, we have that with probability at least $1-\d$ for all $j_1^*$, $i$, and $\tau_1, \tau_2$ we have 
    \begin{equation*}
        \sum_{t = \tau_1}^{\tau_2} \One{i_t = i} \tilde r_t^i(j^*)
        \le
        \sum_{t = \tau_1}^{\tau_2} q_i \tilde r_t^i(j^*)
        +
        \order{ U_{\tau_2} \sqrt{T \log\qty(T K N /\d)} }
    \end{equation*}
    
    Substituting $N$ and $K$, the above makes \eqref{eq:poly:61}
    \begin{alignat}{3}\label{eq:poly:62}
        \Line{
            \sup_{f \in \calF} \sum_{t = \tau_1}^{\tau_2} r_t(f)
        }{\le}{
            \sum_{i\in[N]} q_i \max_{j^* \in [K]} \sum_{t = \tau_1}^{\tau_2} \tilde r_t^i(j^*)
            +
            \order{ U_{\tau_2} \sqrt{T \log\qty(T L /\d)} }
        }{}
    \end{alignat}
    
    We now bound $q_i \max_{j^*} \sum_{t = \tau_1}^{\tau_2} \tilde r_t^i(j^*)$ for some $i$.
    By \eqref{eq:poly:60} we have
    \begin{alignat*}{3}
        \Line{
            q_i \max_{j^*} \sum_{t = \tau_1}^{\tau_2} \tilde r_t^i(j^*)
        }{\le}{
            q_i \sum_{t = \tau_1}^{\tau_2} \tilde r_t^i(j_t^i)
            +
            q_i U_{\tau_2} \order{ \sqrt{T \log (T L / \d)} }
        }{}
        \\
        \Line{}{\le}{
            \sum_{t = \tau_1}^{\tau_2} \One{i_t = i} \tilde r_t^i(j_t^i)
            +
            q_i U_{\tau_2} \order{ \sqrt{T \log (T K N / \d)} }
            +
            q_i U_{\tau_2} \order{ \sqrt{T \log (T L / \d)} }
        }{}
        \\
        \Line{}{=}{
            \sum_{t = \tau_1}^{\tau_2} \One{i_t = i} \tilde r_t^i(j_t^i)
            +
            q_i U_{\tau_2} \order{ \sqrt{T \log (T L / \d)} }
        }{}
    \end{alignat*}
    where in the last inequality we used Azuma's inequality, which heavily depends on the fact that $i_t$ and $\tilde r_t^i(j_t^i)$ are independent conditioned on the history of rounds up to $t-1$, since $\tilde r_t^i(j_t^i)$ does not depend on $v_t$ but only $d_t$.
    Plugging the above into \eqref{eq:poly:62} we get
    \begin{alignat*}{3}
        \Line{
            \sup_{f \in \calF} \sum_{t = \tau_1}^{\tau_2} r_t(f)
        }{\le}{
            \sum_{i\in[N]} \sum_{t = \tau_1}^{\tau_2} \One{i_t = i} \tilde r_t^i(j_t^i)
            +
            U_{\tau_2} \order{ \sqrt{T \log (T L / \d)} }
        }{}
    \end{alignat*}
    
    By noticing that $\tilde r_t^i(j_t) \le r_t(\tilde b_{t, j_t}) + U_t(1/K + 1/N)$ and bounding $T(1/N + 1/K) = O(1)$ we get the $U_{\tau_2} \order{ \sqrt{T \log (T L / \d)} }$ high probability interval regret bound on the rewards $r_t(\cdot)$.
    Using \cref{thm:pre:primal_dual} we get the desired result.
    
    The tight satisfaction of the ROI constraint follows by using \cref{lem:tight:reduction}.
    The first algorithm is the one we describe above.
    The second algorithm is the primal algorithm described above with $\chi_t = \psi_t = 1$, which achieves $Q(\tau, \d) = \tau \b - \order*{ \sqrt{\tau \log (\tau L / \d)} }$.
\end{proof}